\documentclass{lmcs}
\pdfoutput=1

\usepackage{lastpage}
\lmcsdoi{15}{4}{15}
\lmcsheading{}{\pageref{LastPage}}{}{}%
{Apr.~09,~2019}{Dec.~18,~2019}{}


\usepackage[english]{babel}
\usepackage[utf8]{inputenc}
\usepackage[T1]{fontenc}

\usepackage{tikz}
\usetikzlibrary{positioning,shapes,petri}
\usetikzlibrary{automata}
\usetikzlibrary{arrows,backgrounds,positioning,fit,calc}
\usetikzlibrary{matrix}
\usetikzlibrary{decorations.pathreplacing}
\usetikzlibrary{fit}
\usepackage{extarrows}

\newenvironment{theorem}{\begin{thm}}{\end{thm}}
\newenvironment{lemma}{\begin{lem}}{\end{lem}}
\newenvironment{corollary}{\begin{cor}}{\end{cor}}
\newenvironment{proposition}{\begin{prop}}{\end{prop}}
\newenvironment{claim}{\begin{clm}}{\end{clm}}
\newenvironment{definition}{\begin{defi}}{\end{defi}}
\newenvironment{example}{\begin{exa}}{\end{exa}}
\newenvironment{remark}{\begin{rem}}{\end{rem}}
 
\renewcommand{\vec}[1]{\mathbf{#1}}\newcommand{\V}[2]{\bigl(\begin{array}{@{}c@{}}\scriptstyle #1\\[-.5em] \scriptstyle #2\end{array}\bigr)}\newcommand{\T}[3]{\Bigl(\begin{array}{@{}c@{}}\scriptstyle #1\\[-.5em] \scriptstyle #2\\[-.5em] \scriptstyle #3\end{array}\Bigr)}\newcommand{\Ts}[3]{\Bigl(\begin{array}{@{}c@{}}\scriptstyle #1\\[-.4em] \scriptstyle #2\\[-.5em] \scriptstyle #3\end{array}\Bigr)}\newcommand{\Nat}{{\mathbb{N}}}
\newcommand{\Zed}{{\mathbb{Z}}}
\newcommand{\egdef}{\stackrel{\mbox{\begin{tiny}def\end{tiny}}}{=}}
\newcommand{\equivdef}{\stackrel{\mbox{\begin{tiny}def\end{tiny}}}{\Leftrightarrow}}

\newcommand{\vzero}{{\vec{0}}}
\newcommand{\vunit}{{\vec{i}}}
\newcommand{\va}{{\vec{a}}}
\newcommand{\vb}{{\vec{b}}}
\newcommand{\vc}{{\vec{c}}}
\newcommand{\vd}{{\vec{d}}}
\newcommand{\ve}{{\vec{e}}}
\newcommand{\vp}{{\vec{p}}}

\newcommand{\vx}{{\vec{x}}}
\newcommand{\vy}{{\vec{y}}}

\newcommand{\Pos}{{\mathit{Pos}}}

\newcommand{\setN}{\Nat}
\newcommand{\setZ}{\Zed}

\newcommand{\tuple}[1]{\langle {#1} \rangle}

\renewcommand{\root}{\operatorname{root}}
\newcommand{\child}{\operatorname{adorn}}

\newcommand{\dire}[1]{\operatorname{dir}(#1)}

\newcommand{\step}{\xLongrightarrow{}}
\newcommand{\stepstar}{\xLongrightarrow{*}}

\newcommand{\NTfont}[1]{{\mathtt{#1}}}\newcommand{\NTF}{\NTfont{F}}
\newcommand{\NTS}{\NTfont{S}}
\newcommand{\REC}{\NTfont{Rec}}
\newcommand{\LIM}{\NTfont{Lim}}
\newcommand{\REST}{\NTfont{Pop}}

\newcommand{\FO}{\mathsf{FO}}

\begin{document}

\title[On Functions Weakly Computable by Pushdown Petri Nets]{On Functions Weakly Computable by Pushdown Petri Nets and Related Systems}

\author[J.~Leroux]{Jérôme Leroux}
\address{LaBRI, Univ.\ Bordeaux \& CNRS, France}
\thanks{  This work was partly supported by grant ANR-17-CE40-0028 of
the French National Research Agency ANR (project BRAVAS), and by the Indo-French CNRS UMI 2000 ReLaX.
}

\author[M.~Praveen]{M Praveen}
\address{Chennai Mathematical Institute, India}

\author[Ph.~Schnoebelen]{Philippe Schnoebelen}
\address{LSV, ENS Paris-Saclay \& CNRS, France}

\author[G.~Sutre]{Grégoire Sutre}
\address{LaBRI, Univ.\ Bordeaux \& CNRS, France}

\begin{abstract}
We consider numerical functions weakly computable by
grammar-controlled vector addition systems (GVASes, a variant of
 pushdown Petri nets).
GVASes can weakly compute all fast growing functions $F_\alpha$ for
$\alpha<\omega^\omega$, hence they are computationally more powerful than
standard vector addition systems.
On the other hand they cannot weakly compute the inverses $F_\alpha^{-1}$
or indeed any  sublinear function.
The proof relies on a pumping lemma for runs of GVASes that is of independent interest.
\end{abstract}

\maketitle
 
\section{Introduction}
\label{sec-intro}

Pushdown Petri nets are Petri nets extended with a pushdown
stack.
They have been used to model asynchronous programs~\cite{SenV06} and,
more generally,
recursive programs with integer variables~\cite{AtigG11}.
They sometimes appear under a different but
essentially equivalent guise:
stack/pushdown/context-free vector addition systems~\cite{lazic2013b,leroux2014,leroux2015d},
partially blind multi-counter machines~\cite{Greibach78a} with a pushdown stack,
etc.
It is not yet known whether reachability is decidable for pushdown
Petri nets and this is one of the major open problems in computer
science.  However, a series of recent results improved our
understanding of the computational power of these models:
coverability, reachability and boundedness are \textsc{Tower}-hard~\cite{lazic2013b,lazic2017},
and boundedness is solvable in hyper-Ackermannian time~\cite{leroux2014}.

\smallskip

With the present article, we contribute to this line of work.  We
recall \emph{Grammar-Controlled Vector Addition Systems}~\cite{leroux2015d},
or GVAS,
a variant model, close to Pushdown Petri nets, where the pushdown
stack is replaced by a context-free restriction on the firing of
rules. The runs are now naturally organized in a derivation tree, and
the stack is not actually present in the configurations: this leads to
a simplified mathematical treatment, where the usual monotonicity
properties of VASes can be put to use.

\bigskip

As a step towards understanding the expressive power of these GVASes,
we consider the number-theoretical functions that are \emph{weakly
computable} in this model. Restricting to weakly computing a numerical
function is a natural idea when dealing with models like VASes and
GVASes that lack zero-tests, or, more precisely, that cannot initiate
a given action on the condition that a counter's value is zero, only
on the condition that it is not zero.

\smallskip

This notion has been used since the early days of Petri nets and has
proved very useful in hardness or impossibility proofs: For Petri nets
and VASSes, the undecidability of equivalence problems, and the
Ackermann-hardness of the same problems for bounded systems, have been
proved using the fact that multivariate polynomials with positive
integer coefficients ---aka positive Diophantine polynomials--- and,
respectively, the fast-growing functions $(F_i)_{i\in\Nat}$ in the
Grzegorczyk hierarchy, are all weakly
computable~\cite{hack76b,Mayr-Meyer81,jancar2001}.
More recently, the nonelementary complexity lower bound for VASS
reachability is obtained thanks to a uniform (polynomial size) family
of systems
computing (exactly) $n$-EXP(2) from $n$~\cite{leroux-stoc19}.

\smallskip

The above results rely on showing how \emph{some} useful functions are
weakly computable by Petri nets and VASSes.
But not much is known about exactly which functions are weakly computable or not.
It is known that all such functions are monotonic. They are all
primitive-recursive. The class of weakly computable functions is closed under
composition.

In this article, we show that functions weakly computable by GVASes go
beyond  those weakly computable by VASSes, in particular we show how
to weakly compute the Fast Growing $(F_\alpha)$ for all $\alpha<\omega^\omega$.

\smallskip

A folklore conjecture states that the inverses of the
fast-growing functions are not weakly computable by Petri nets. It is stated as fact
in~\cite[p.252]{phs-IPL2002} but no reference is given. In this article, 
we settle the issue by proving that any unbounded function weakly computable by
Petri nets and more generally by GVASes
is in $\Omega(x)$, i.e., it eventually dominates
$c\cdot x$ for some constant $c>0$. Thus any function that is
sublinear, like $x\mapsto \lfloor \sqrt{x} \rfloor$, or $x\mapsto
\lfloor \log x \rfloor$, is  not weakly computable by GVASes.
The proof technique is interesting in its own right: it relies on a
pumping lemma on runs of GVASes that could have wider applications.
This pumping lemma follows from a well-quasi-ordering on the set of runs
that further directs it.

\subsection*{Beyond Petri nets and VASSes.}
Petri nets and VASSes are a classic example of well-structured
systems~\cite{abdulla2000c,finkel2001}. In recent years, weakly
computing numerical functions has proved to be a fundamental tool for
understanding the expressive power and the complexity of some families
of well-structured systems that are more powerful than Petri nets and
VASSes~\cite{phs-mfcs2010,HSS-lics2012,HSS-lmcs}.
For such systems,
the hardness proofs rely on weakly computing fast-growing
functions $(F_\alpha)_{\alpha\in\textit{Ord}}$ that extend
Grzegorczyk's hierarchy. These hardness proofs also crucially  rely on
weakly computing the inverses of the $F_\alpha$'s.

\smallskip

There are several extensions of Petri nets for which reachability (or
coverability or boundedness) remains decidable: nets with nested
zero-tests~\cite{Reinhardt08}, recursive VASSes~\cite{bouajjani2012}
and Branching VASSes~\cite{demri2013}, VASSes with pointers to
counters~\cite{demri2013c}, unordered data Petri nets~\cite{datanetsfull08},
etc., and of course pushdown VASes and GVASes.
For the latter, while coverability and reachability are still
open in general, partial decidability results have been
obtained by looking at sub-classes,
namely GVASes with finite-index grammars~\cite{AtigG11} and
GVASes of dimension one~\cite{leroux2015d}.
In many cases, it is not known how these extensions compare
in expressive power and in complexity. We believe that weakly
computable functions can be a useful tool when addressing these
questions.

\subsection*{Related models.}
The GVAS model can simulate counter machines extended with nested
zero-tests (from~\cite{Reinhardt08}), and the vector addition systems extended with a pushdown
stack (from~\cite{leroux2014}). The first simulation
was shown in~\cite{AtigG11} and holds even for GVASes with finite-index grammars.
The second
one comes from the classical transformation of a pushdown automaton
into a context-free grammar that recognizes the same language.
There exists still other models that extend vector addition systems with
stack-related mechanisms, e.g., Mayr's Process Rewrite
Systems~\cite{mayr2000} or Haddad and Poitrenaud's Recursive Petri
Nets~\cite{haddad2007}. Pending some further, more formal, comparison,
it seems that these models are less expressive than
Pushdown VASes since they only allow limited interactions between
stack and counters.

\subsection*{Outline of the paper.}
Section~\ref{sec-gvas} introduces GVASes and
fixes some notation. 
In Section~\ref{sec-pumping}, we introduce flow trees, a tree-shaped
version of runs of GVASes for which we develop our two main tools: a
well-quasi-ordering between flow trees and an Amalgamation Theorem. 
The following two sections explore applications of the Amalgamation
Theorem
in understanding the computing power of GVASes:
via GVAS-definable sets in Section~\ref{sec-gvas-sets},
via weakly computable function in Section~\ref{sec-weakcomp}.
Finally, we show in Section~\ref{sec-hypack} that GVASes can weakly
compute all Fast-Growing functions $F_\alpha$ for $\alpha<\omega^\omega$.

 \section{Grammar-Controlled Vector Addition Systems}
\label{sec-gvas}

This section recalls the model of grammar-controlled vector addition
systems, originally from~\cite{leroux2015d}.
In a nutshell,
these are intersections of classical VAS with context-free grammars.
Remark~\ref{rem:pvas} relates them with the equivalent model of pushdown
vector addition systems.

\subsection*{Vector Addition Systems}
For a \emph{dimension} $d\in\Nat$, we consider \emph{configurations}
that are vectors $\vc,\vd,\vx,\vy,\ldots$ in $\Nat^d$, and
\emph{actions} that are vectors $\vec{a}\in\setZ^d$.
We write $\vx\xrightarrow{\va}\vy$ for two configurations $\vx,\vy$ in
$\Nat^d$ if $\vy=\vx+\va$.
A \emph{vector addition system} (a \emph{VAS}) is a transition system
of the form $(\Nat^d,\{\xrightarrow{\va}\}_{\va\in\vec{A}})$ generated by a
finite set $\vec{A}\subseteq\setZ^d$ of actions.

\smallskip

In a VAS, the one-step transition relations
$\{\xrightarrow{\va}\}_{\va\in \vec{A}}$ are composed in a natural way: with
any word $w=\vec{a}_1\cdots \va_k\in \vec{A}^*$ of actions, we associate
the binary relation $\xrightarrow{w}$ defined over configurations by
$\vec{x}\xrightarrow{w}\vec{y}$ iff there exists a sequence
$\vec{c}_0,\ldots,\vec{c}_k$ of configurations such that
$\vec{c}_0=\vec{x}$, $\vec{c}_k=\vec{y}$ and such that
$\vec{c}_{j-1}\xrightarrow{\va_j}\vec{c}_{j}$ for every $1\leq j\leq k$.  Those
relations are monotonic:
\begin{equation}\label{eq-vas-monotonic}
  \vec{x}\xrightarrow{w}\vec{y} \text{ and }\vec{v}\in\Nat^d\text{
    implies }\vec{x}+\vec{v}\xrightarrow{w}\vec{y}+\vec{v} \:.
\end{equation}

\subsection*{Notation}
When writing configurations $\vc\in\Nat^d$, we sometimes split the
vector in parts, writing e.g., $\vc=(\vx,\vy)$ for some
$\vx\in\Nat^{d_1}$ and $\vy\in\Nat^{d_2}$ with $d=d_1+d_2$. We also
write $\vzero_d$ for the null vector in $\Nat^d$, often leaving the
dimension implicit.

\subsection*{Grammar-controlled Vector Addition Systems}
A $d$-dimensional \emph{grammar-controlled vector addition system} (a
\emph{GVAS}) can be seen as a context-free grammar using terminals from
$\setZ^d$,
or equivalently as a VAS where the valid sequences of actions are generated by a
context-free grammar.
Formally, a $d$-dimensional GVAS is some $G=(V,\vec{A},R,S)$ where $V$ is a finite set of
\emph{nonterminals}, where $\vec{A}\subseteq
\setZ^d$ is a finite set of terminals called \emph{actions}, where $R\subseteq
V\times (V\cup \vec{A})^*$ is a finite set of production \emph{rules}, and
$S\in V$ is the \emph{start symbol}.
Following the usual convention, a rule $(T,u)$ is also written $T \rightarrow u$.
We denote nonterminals from $V$ with capital letters like $S,T,\ldots$
while symbols from the larger set $V\cup \vec{A}$ are denoted with
$X,Y,\ldots$.  Words in $(V\cup \vec{A})^*$ are denoted with
$w,u,v,\ldots$. As usual, $\varepsilon$ denotes the empty word.

\smallskip

For all words $w,w'\in (V\cup \vec{A})^*$, we say that $w\step w'$ is
a \emph{derivation step} of $G$ if there exist two words $v,v'$ in
$(V\cup \vec{A})^*$ and a rule $(T,u)$ in $R$ such that $w=vTv'$
and $w'=vuv'$. Let $\stepstar$ denote the reflexive and
transitive closure of $\step$.
The language $L_G\subseteq \vec{A}^*$ generated by $G$ seen as a grammar is
defined as usual with $w\in L_G\equivdef S\stepstar w\in
\vec{A}^*$. More generally, for any $u\in (V\cup \vec{A})^*$, the language $L_G(u)$ is $\{w\in
\vec{A}^*~|~u\stepstar w\}$.

\smallskip

When $G$ is a GVAS, we are interested in what
sequences of actions may occur between configurations in $\Nat^d$. For this, we
extend the definition of the  $\xrightarrow{w}$ relation
and consider $\xrightarrow{u}$ for any $u\in (V\cup
\vec{A})^*$. Formally, we let
\begin{gather}
\vec{x}\xrightarrow{u}\vec{y}
\equivdef
\exists w\in L_G(u):\vec{x}\xrightarrow{w}\vec{y}\:.
\end{gather}
A labeled pair $\vx\xrightarrow{u}\vy$ is called a \emph{run} of the
GVAS, and should not be confused with the derivations $w\stepstar w'$ that
only involve the grammar part.

Like VASes, GVASes are monotonic:
\begin{equation}\label{eq-gvas-monotonic}
  \vec{x}\xrightarrow{u}\vec{y} \text{ and }\vec{v}\in\Nat^d\text{
    implies }\vec{x}+\vec{v}\xrightarrow{u}\vec{y}+\vec{v}\:.
\end{equation}
The underlying grammar $G$ is left implicit in the above notation.
We sometimes write $\vx\xrightarrow{G}\vy$ instead of
 $\vx\xrightarrow{S}\vy$, where $S$ is the start symbol of $G$, when
several grammars are considered simultaneously.

\begin{example}
\label{ex-power-2}
Let $d=1$, $V=\{S,T\}$, and consider the
1-dimensional GVAS
given by the following four rules in Backus-Naur form:
\begin{xalignat*}{2}
S & \rightarrow \vec{1} ~\big|~ \vec{-1} \, S \, T       \:,
&
T & \rightarrow \vec{0} ~\big|~ \vec{-1} \, T \, \vec{2} \:.
\end{xalignat*}

Since we shall claim in Section~\ref{sec-weakcomp} that this GVAS
weakly computes the $2^n$ function, let us state and prove the
main properties of its runs. Formally, for every $k,k',n,n'\in\Nat$,
one has:
\begin{align}
\label{eq-spec-power-2}
  k\xrightarrow{T}k' & \text{ iff } k\leq k'\leq 2k \:,
&
  n\xrightarrow{S}n' & \text{ iff } 1\leq n'\leq 2^n \:.
\end{align}
Thus, even in dimension 1, the reachability relation may not be
semilinear.

To prove \eqref{eq-spec-power-2}
assume first that $k\xrightarrow{T}k'$ for some natural numbers
$k,k'$. There exists $m\in\setN$ such that
$k\xrightarrow{\vec{-1}^m\vec{0}\vec{2}^m}k'$. In particular $m\leq k$
and $k'=k+m$. We deduce that $k\leq k'\leq 2k$. Conversely,
let $k,k'$ be two natural numbers such that $k\leq k'\leq
2k$. Observe that $T\stepstar \vec{-1}^n\vec{0}\vec{2}^n$
where $n$ is defined as $k'-k$. The following relations show that
$k\xrightarrow{T}k'$:
$$k ~\xrightarrow{\vec{-1}^n}~ k-n ~\xrightarrow{\vec{0}}~
k-n ~\xrightarrow{\vec{2}^n}~ k'\:.$$

Now, assume that $n\xrightarrow{S}n'$ for some natural numbers $n,n'$,
There exists $m\in\setN$ such that
$n\xrightarrow{\vec{-1}^m\vec{1}T^m}n'$. It follows that $m\leq n$,
and from the previous paragraph, we deduce that $n'\leq (n-m+1)2^m\leq
2^n$ by observing that $x+1\leq 2^x$ for every $x\in\Nat$ and by
replacing $x$ by $n-m$. Conversely, let $n,n'$ be two natural numbers such that $1\leq
n'\leq 2^n$. Observe that $S\stepstar
\vec{-1}^n\vec{1}T^n$. Let us introduce a natural number $m$ in $\{0,\ldots,n-1\}$ such that $2^m\leq
n'\leq 2^{m+1}$. The following relations show that $n\xrightarrow{S}n'$:
\begin{equation}
\tag*{$\mathqed$}
n\xrightarrow{\vec{-1}^n}0\xrightarrow{\vec{1}}2^0\xrightarrow{T}2^1\cdots
\xrightarrow{T}2^{m}\xrightarrow{T}n'\xrightarrow{T^{n-1-m}}n'\:.
\end{equation}
\end{example}

\begin{example}
\label{ex-gvas}
Let $G$ be the 2-dimensional GVAS with a single nonterminal symbol and the
following three rules:
\[
       S \rightarrow S\,S ~\big|~ \V{-1}{2} ~\big|~ \V{2}{-1}
\:.
\]
Let $w=\V{-1}{2}\V{2}{-1}\V{-1}{2}$ and observe that $S \stepstar w$. We
have $\V{2}{2} \xrightarrow{\V{-1}{2}} \V{1}{4}
\xrightarrow{\V{2}{-1}} \V{3}{3} \xrightarrow{\V{-1}{2}}
\V{2}{5}$ and hence $\V{2}{2} \xrightarrow{S} \V{2}{5}$.
\qed
\end{example}

\begin{remark}
  \label{rem:pvas}
  It is well known that, from a formal language viewpoint,
  context-free grammars are equivalent to pushdown automata.
  Similarly,
  GVAS can be equivalently presented as VAS extended with a pushdown stack.
  Formally,
  a $d$-dimensional \emph{Pushdown Vector Addition System} (a \emph{PVAS})
  is a transition system of the form
  $(\Nat^d \times \Gamma^*, \{\xrightarrow{\vp}\}_{\vp\in\vec{P}})$
  generated by a pair $(\Gamma, \vec{P})$ where $\Gamma$ is a finite
  \emph{stack alphabet} and
  $\vec{P} \subseteq \Gamma^* \times \Gamma^* \times \Zed^{d}$ is a finite
  set of \emph{actions}.
  So configurations are now pairs $(\vec{x}, u)$ where
  $\vec{x} \in \Nat^{d}$ is as for VAS
  and $u \in \Gamma^*$ is a word denoting the contents of the stack.
  Intuitively, an action $\vp = (\alpha, \beta,\vec{a})$ pops the
  string $\alpha$ from the top of the stack, then pushes the string
  $\beta$ onto the top of the stack, and adds $\vec{a}$ to the vector
  of natural numbers.
  Formally,
  each action $(\alpha, \beta, \vec{a}) \in \vec{P}$ induces a binary
  relation $\xrightarrow{(\alpha,\beta,\vec{a})}$ on configurations
  defined by
  $(\vec{x}, u)
  \xrightarrow{(\alpha,\beta,\vec{a})}
  (\vec{y}, v)$
  if $\vec{y} = \vec{x} + \vec{a}$ and there exists $w$ such that
  $u = \alpha w$ and $v = \beta w$.
  GVASes can be translated into equivalent PVASes and vice-versa.

  \smallskip

  For instance,
  the PVAS corresponding to Example~\ref{ex-gvas} is generated by the
  pair $(\Gamma, \vec{P})$ where $\Gamma = \{S\}$ and $\vec{P}$ is the
  set of actions
  $\{(S, SS, \V{0}{0}), (S, \varepsilon, \V{-1}{2}),
  (S, \varepsilon, \V{2}{-1})\}$.
  Corresponding to $\V{2}{2} \xrightarrow{w} \V{2}{5}$ with
  $w=\V{-1}{2}\V{2}{-1}\V{-1}{2}$ there, we have 
\[
(S,\V{2}{2})
  \xrightarrow{\! S,SS,\V{0}{0}\!} (SS,\V{2}{2}) 
  \xrightarrow{\! S, \varepsilon, \V{-1}{2}\!} (S, \V{1}{4}) 
  \xrightarrow{\! S,SS,\V{0}{0}\!} (SS, \V{1}{4}) 
  \xrightarrow{\! S, \varepsilon, \V{2}{-1}\!} (S,\V{3}{3})
  \xrightarrow{\! S, \varepsilon, \V{-1}{2}\!} (\varepsilon,\V{2}{5})
\]
  in the PVAS.
\end{remark}

 \section{Well-Quasi-Ordering Runs in GVASes}
\label{sec-pumping}

In this section we define flow trees of GVASes and show that they
satisfy an amalgamation property. This property is used in the next
section to provide a geometrical decomposition of GVAS sets, and in
the following section to show that unbounded weakly computable
functions are in $\Omega(n)$.

\medskip

Let $G=(V,\vec{A},R,S)$ be a $d$-dimensional GVAS.  \emph{Flow trees} of
$G$ are trees that combine a transition $\vx\xrightarrow{w}\vy$ in the
VAS part of $G$ with a derivation tree for the corresponding
$S\stepstar w$ in the grammar part of $G$.

Flow trees are finite rooted ordered trees labeled with transitions of
$G$: we write $t = \sigma[t_1,\ldots,t_\ell]$ to denote a flow tree
$t$ made of a root with $\ell$ subtrees $t_1,\ldots,t_\ell$. The root
is labeled by a transition $\sigma$ of $G$, of the form
$\vc\xrightarrow{X}\vd$ with $X\in V\cup \vec{A}$.
We write $\root(t) = \sigma$.
Formally, $F(G)$ is the
least set of trees that contains all $(\vc\xrightarrow{\va}\vd)[]$ with
$\va\in \vec{A}$ and $\vc+\va=\vd$, and all
$(\vc\xrightarrow{T}\vd)[t_1,\ldots,t_\ell]$ with $T\in V$ and
$t_1,\ldots,t_\ell\in F(G)$ such that there is a rule $T\step
X_1\cdots X_\ell$ in $R$ and configurations
$\vc_0,\vc_1,\ldots,\vc_\ell$ with $\vc_0=\vc$, $\vc_\ell=\vd$ and
such that, for $i=1,\ldots,\ell$, the root of $t_i$ is labeled with
$\vc_{i-1}\xrightarrow{X_i}\vc_i$.  A \emph{subtree} of
$t=\sigma[t_1,\ldots,t_\ell]$ is either $t$ itself or a subtree of
some $t_i$ for $i=1,\ldots,\ell$.  A (sub)tree
$(\vc\xrightarrow{X}\vd)[t_1,\ldots,t_\ell]$ is a \emph{leaf} when
$\ell=0$: this requires that $X=\va\in \vec{A}$ is an action (and then
$\vd=\vc+\va$) or that $X=T\in V$ is a non-terminal and
$T\step\varepsilon$ is a rule in $R$ (and then $\vd=\vc$).

As is standard, we use \emph{positions} to identify occurrences of
subtrees inside $t$. Formally, a position is a finite sequence of
natural numbers, and the positions of the subtrees of $t$, denoted
$\Pos(t)$ are given inductively by
\[
\Pos(\sigma[t_1,\ldots,t_l]) \egdef \{\varepsilon\}\cup\{i.q~|~1\leq i\leq l \land
q\in\Pos(t_i)\} \:.
\]
For $p\in\Pos(t)$, the subtree of $t$ at position $p$ is denoted
$t/p$.

\begin{example}
\label{ex-flow-tree}
Recall the 1-dimensional GVAS $G$ from Example~\ref{ex-power-2}.
The grammar admits, among others, a derivation $S\stepstar w$ for
$w=\vec{-1}\,\vec{-1}\,\vec{1}\,\vec{0}\,\vec{0}$.
Thus $3\xrightarrow{w} 2$ is a transition in $G$.

In Fig.~\ref{fig-ex-flowtree2} we display a derivation tree witnessing
$S\stepstar w$ and a flow tree witnessing $\vc\xrightarrow{w}\vd$ for
$\vc=3$ and $\vd=2$.
\qed
\end{example}
\begin{figure}[t!]
  \begin{center}
\begin{tikzpicture}[auto,node distance=2em]     \node (a)              {$S$};
    \node (a1) [below=of a,xshift=-3em] {$\vec{-1}$};
    \node (a2) [below=of a] {$S$};
    \node (a3) [below=of a,xshift=4em] {$T$};
    \node (a21) [below=of a2,xshift=-2em] {$\vec{-1}$};
    \node (a22) [below=of a2] {$S$};
    \node (a23) [below=of a2,xshift=2em] {$T$};
    \node (a31) [below=of a3] {$\vec{0}$};
    \node (l23) [below=1em of a23] {$\vec{0}$};
    \node (l22) [below=1em of a22] {$\vec{1}$};

    \draw[-,thick] (a) to (a1);
    \draw[-,thick] (a) to (a2);
    \draw[-,thick] (a) to (a3);
    \draw[-,thick] (a2) to (a21);
    \draw[-,thick] (a2) to (a22);
    \draw[-,thick] (a2) to (a23);
    \draw[-,thick] (a3) to (a31);
    \draw[-,thick] (a22) to (l22);
    \draw[-,thick] (a23) to (l23);

    \node (A)     [right=15em of a]          {$3\xrightarrow{S}2$};
    \node (A1) [below=of A,xshift=-6em] {$3\xrightarrow{\vec{-1}}2$};
    \node (A2) [below=of A, xshift=-1em] {$2\xrightarrow{S}2$};
    \node (A3) [below=of A,xshift=7em] {$2\xrightarrow{T}2$};
    \node (A21) [below=of A2,xshift=-4em] {$2\xrightarrow{\vec{-1}}1$};
    \node (A22) [below=of A2] {$1\xrightarrow{S}2$};
    \node (A23) [below=of A2,xshift=4em] {$2\xrightarrow{T}2$};
    \node (A31) [below=of A3] {$2\xrightarrow{\vec{0}}2$};
    \node (L23) [below=1em of A23] {$2\xrightarrow{\vec{0}}2$};
    \node (L22) [below=1em of A22] {$1\xrightarrow{\vec{1}}2$};

    \draw[-,thick] (A) to (A1);
    \draw[-,thick] (A) to (A2);
    \draw[-,thick] (A) to (A3);
    \draw[-,thick] (A2) to (A21);
    \draw[-,thick] (A2) to (A22);
    \draw[-,thick] (A2) to (A23);
    \draw[-,thick] (A3) to (A31);
    \draw[-,thick] (A22) to (L22);
    \draw[-,thick] (A23) to (L23);
    \end{tikzpicture}
\end{center}
\caption[caption]{Continuing Example~\ref{ex-power-2}: $S \rightarrow \vec{1} ~|~ \vec{-1} \, S \, T$ and $T \rightarrow \vec{0} ~|~ \vec{-1} \, T \, \vec{2}$.\\\hspace{\textwidth}
Left: Derivation tree witnessing $S\stepstar w$, with $w=\vec{-1} \vec{-1}\,\vec{1}\,\vec{0}\,\vec{0}$.\\\hspace{\textwidth} 
Right: Flow tree witnessing $3\xrightarrow{w}2$.}
\label{fig-ex-flowtree2}
\end{figure}

\medskip

We now extend to flow trees of GVASes an ordering
initially introduced by Jan\v{c}ar for runs of VASes~\cite[Def.~6.4]{jancar90}.

\begin{definition}[Ordering GVAS transitions and flow trees]
\label{def-leq_G}
For two transitions $\sigma=\vc\xrightarrow{X}\vd$ and
$\theta=\vc'\xrightarrow{X'}\vd'$ with $X,X'\in V\cup \vec{A}$, we let
\[
\sigma\leq\theta \:\equivdef\: \vc\leq \vc'\land \vd\leq\vd'\land X=X' \:.
\]

The ordering $\leq_G$ between flow trees $s,t\in F(G)$ is defined
by induction on the structure of trees:
$s=\sigma[s_1,\ldots,s_k]\leq_G t=\theta[t_1,\ldots,t_\ell]$ if, and
only if, $\sigma\leq \theta$ and there exists a subtree $t'$ of $t$
of the form $t'=\theta'[t'_1,\ldots,t'_{\ell'}]$ with $\sigma\leq
\theta'$, $\ell'=k$ and $s_j\leq_G t'_j$ for every $1\leq j\leq k$.
\end{definition}
This definition is well-founded and, since the subtree relation is
transitive, $\leq_G$ is clearly reflexive and transitive, i.e., is a
quasi-ordering. 
In the appendix, we prove the following key property:
\begin{lemma}[See Appendix~\ref{app-wqo}]
\label{lem-FG-wqo}
$(F(G),\leq_G)$ is a well-quasi-ordering.
\end{lemma}
In other words, any infinite sequence $s_0,s_1,s_2,\ldots$ of flow
trees contains an infinite increasing subsequence
$s_{i_0}\leq_G s_{i_1}\leq_G s_{i_2}\leq_G \cdots$.
\begin{example}
\label{ex:flow-tree-ordering}
The flow trees shown in Fig.~\ref{fig:ordering} illustrate the
ordering $\le_G$, on a $1$-dimensional GVAS with non-terminals
$S,T,U,V$ and the following rules: $S \rightarrow \vec{3} \, T ~|~
\vec{3} \, U$, $T \rightarrow \vec{-2} ~|~ V \, T$,
$U \rightarrow T$ and $V \rightarrow \epsilon$. We can see that $t_1
\leq_G t_2$ from the following orderings on the subtrees of $t_1$ and
$t_2$: $t_1/1 \leq_G t_2/1$, $t_1/2 \leq_G t_2/2$ and $t_1/21 \leq_G
t_2/221$. It can
be verified that $t_1 \not\leq_G t_0$. However, $t_1 \sqsubseteq t_0$,
where $\sqsubseteq$ is the standard homeomorphic embedding~\footnote{  defined by $s = \sigma[s_1,\ldots,s_k] \sqsubseteq t$ if, and only if,
  there exists a subtree $t' = \theta'[t'_1,\ldots,t'_{\ell'}]$ of $t$
  such that $\sigma \leq \theta'$, $k \leq \ell'$ and
  $s_1 \sqsubseteq t'_{j_1}, \ldots, s_k \sqsubseteq t'_{j_k}$ for some
  subsequence $1 \leq j_1 < \cdots < j_k \leq \ell'$.
} of labeled
trees. This can be seen by the following orderings: $t_1/1 \sqsubseteq
t_0/1$, $t_1/2 \sqsubseteq t_0/21$ and $t_1/21 \sqsubseteq t_0/211$.
\begin{figure}[t!]
  \begin{center}
    \scalebox{0.95}{
	\begin{tikzpicture}[auto,node distance=2em]       	    \node (c) {$3 \xrightarrow{S} 4$};
            \node (c1) [below=of c,xshift=-3em] {$3 \xrightarrow{\vec{3}} 6$};
	    \node (c2) [below=of c,xshift=3em] {$6 \xrightarrow{U} 4$};
	    \node (c21) [below=of c2] {$6 \xrightarrow{T} 4$};
            \node (c211) [below=of c21] {$6 \xrightarrow{\vec{-2}} 4$};

           \node at ([yshift=2em]c) {$t_0$};

           \draw[-,thick] (c) -- (c1);
           \draw[-,thick] (c) -- (c2);
           \draw[-,thick] (c2) -- (c21);
           \draw[-,thick] (c21) -- (c211);

           \node (ca) at ([xshift=6em]c) {\LARGE $\not\geq_G$};

   	    \node (a) at ([xshift=6em]ca) {$2 \xrightarrow{S} 3$};
	    \node (a1) [below=of a,xshift=-3em] {$2 \xrightarrow{\vec{3}} 5$};
	    \node (a2) [below=of a,xshift=3em] {$5 \xrightarrow{T} 3$};
            \node (a21) [below=of a2] {$5 \xrightarrow{\vec{-2}} 3$};

           \node at ([yshift=2em]a) {$t_1$};

           \draw[-,thick] (a) -- (a1);
           \draw[-,thick] (a) -- (a2);
           \draw[-,thick] (a2) -- (a21);

        \node (ab) at ([xshift=6em]a) {\LARGE $\leq_G$};

        	    \node (b) at ([xshift=8em]ab) {$3 \xrightarrow{S} 4$};
            \node (b1) [below=of b,xshift=-4em] {$3 \xrightarrow{\vec{3}} 6$};
	    \node (b2) [below=of b,xshift=4em] {$6 \xrightarrow{T} 4$};
            \node (b21) [below=of b2, xshift=-3em] {$6 \xrightarrow{V} 6$};
            \node (b22) [below=of b2, xshift=3em] {$6 \xrightarrow{T} 4$};
            \node (b221) [below=of b22] {$6 \xrightarrow{\vec{-2}} 4$};

           \node at ([yshift=2em]b) {$t_2$};

           \draw[-,thick] (b) -- (b1);
           \draw[-,thick] (b) -- (b2);
           \draw[-,thick] (b2) -- (b21);
           \draw[-,thick] (b2) -- (b22);
           \draw[-,thick] (b22) -- (b221);

          \begin{scope}[on background layer]
             \node[fill=gray!30, rectangle, rounded corners, fit=(b)] {};
             \node[fill=gray!30, rectangle, rounded corners, fit=(b1)] {};
             \node[fill=gray!30, rectangle, rounded corners, fit=(b2) (b21) (b22)] {};
             \node[fill=gray!30, rectangle, rounded corners, fit=(b221)] {};
           \end{scope}

           \end{tikzpicture}
        }\end{center}
\caption[caption]{Illustration for
    Example~\ref{ex:flow-tree-ordering}: $S \rightarrow \vec{3} \, T ~|~
\vec{3} \, U$, $T \rightarrow \vec{-2} ~|~ V \, T$,
$U \rightarrow T$ and $V \rightarrow \epsilon$.}
\label{fig:ordering}
\end{figure}
 \end{example}

When $\sigma\leq\theta$ for some $\sigma=\vc\xrightarrow{X}\vd$ and
$\theta=\vc'\xrightarrow{X}\vd'$, we also write
$\sigma\leq^\Delta\theta$ with $\Delta=(\vc'-\vc,\vd'-\vd)$. 
Similarly, we write $s\leq_G^\Delta t$ for two flow trees $s$ and $t$
when $s\leq_G t$ and $\root(s)\leq^\Delta\root(t)$.

The pair $\Delta$ is called a \emph{lifting}. 
Note that necessarily
$\Delta$ belongs to $(\Nat^d)^2$ and that $\sigma\leq^\Delta\theta$
and $\theta\leq^{\Delta'}\rho$ entail
$\sigma\leq^{\Delta+\Delta'}\rho$. 
We write $\rho=\sigma+\Delta$ when $\sigma\leq^\Delta\rho$.
Two liftings $\Delta=(\va,\vb)$ and
$\Delta'=(\va',\vb')$ can be chained if $\vb=\va'$. In this case we let
$\Delta\cdot\Delta'\egdef(\va,\vb')$. Note this partial operation is
associative.

When $t/p=t'\leq_G u$ 
we can replace $t'$ by $u$ inside $t$ but this
requires a bit of surgery to ensure the result is well-formed.
First, for a flow tree $t$ and a displacement $\va\in\Nat^d$, 
we let $t+\va$ be the tree defined via
\[
\sigma[t_1,\ldots,t_\ell]+\va
\egdef
(\sigma+(\va,\va))[t_1+\va,\ldots,t_\ell+\va]
\:.
\]
Obviously, $t+\va$ is
a valid flow tree, with $t\leq_G^{(\va,\va)} (t+\va)$.  Now, when
$t/p=t'\leq_G^\Delta u$ for $\Delta=(\va,\vb)$, we define $t[u]_p$ by
induction on $p$ in the following way:
\begin{xalignat*}{2}
t[u]_{\varepsilon} &\egdef u,
&
t[u]_{i.q} &\egdef
(\sigma+\Delta)[t_1+\va,\ldots,t_{i-1}+\va,t_i[u]_q,
t_{i+1}+\vb,\ldots,t_k+\vb].
\end{xalignat*}

\begin{claim}
\label{cl-replace-one-subtree}
If $t/p\leq^\Delta_G u$ then $t[u]_p$ is a valid flow tree
satisfying $t\leq_G^{\Delta}t[u]_p$.
\end{claim}
\begin{proof}
By induction on $p$. If $p=\varepsilon$ the claim holds trivially.
Assume $p=i.q$ with $1\leq i \leq k$ and let $u'=t_i[u]q$. By
induction hypothesis, $t_i\leq_G^\Delta u'$. This implies that
$t[u]_p$ is a well-defined flow tree. Since furthermore $t_j\leq_G
t_j+\va$ when $1\leq j<i$, and symmetrically, $t_j\leq_G t_j+\vb$ when
$i<j\leq k$, we see that $t\leq_G t[u]_p$.  Finally, we observe that
$\root(t)\leq^\Delta\root(t[u]_p)$.
\end{proof}

\begin{lemma}
\label{lem-replace-children}
Let $t=\sigma[t_1,\ldots,t_k]$ and assume $t_i\leq_G^{\Delta_i} u_i$
for $i=1,\ldots,k$. If $\Delta_1\cdot \Delta_2 \cdots \Delta_k=\Delta$
is defined, then $u=(\sigma+\Delta)[u_1,\ldots,u_k]$ is a valid
flow tree satisfying $t\leq_G^\Delta u$.
\end{lemma}
\begin{proof}
Since $\Delta=\Delta_1\cdots \Delta_k$ is defined, we can write
$\Delta_i=(\va_{i-1},\va_i)$ and $\Delta=(\va_0,\va_k)$.  Assume
$\sigma=\vc_0\xrightarrow{X}\vc_k$, with
$\root(t_i)=\vc_{i-1}\xrightarrow{Y_i}\vc_i$ for $i=1,\ldots,k$.  Then
$\root(u_i)=(\vc_{i-1}+\va_{i-1})\xrightarrow{Y_i}(\vc_i+\va_i)$ and
$u$ is a valid transition. That $t\leq_G u$ is immediate.
\end{proof}

\begin{theorem}[Amalgamation]
\label{thm:amalgamation}
If $s\leq_G^{\Delta_1} t_1$ and $s\leq_G^{\Delta_2} t_2$ then there
exists $s'$ s.t.\ $t_1\leq_G^{\Delta_2}s'$ and
$t_2\leq_G^{\Delta_1}s'$ (further entailing
$s\leq_G^{\Delta_1+\Delta_2}s'$).
\end{theorem}
\begin{proof}
By induction on $s$.  Assume $s=\sigma[s_1,\ldots,s_k]$. Since
$s\leq_G^{\Delta_1} t_1$, there is a subtree
$t_1/p=t^1=\rho_1[t_1^1,\ldots,t_k^1]$ of $t_1$ such that
$\sigma\leq\rho_1$ and $s_j\leq_G t_j^1$ for all $j=1,\ldots,k$.
Assume that $\sigma\leq^{\Delta'_1}\rho_1$ and that
$s_j\leq_G^{\Gamma_j}t_j^1$ for $j=1,\ldots,k$. Since $s$ and $t^1$
are valid flow trees, we deduce that $\Delta'_1=\Gamma_1 \cdots
\Gamma_k$.  Symmetrically, from $s\leq_G^{\Delta_2}t_2$, we know that
there is a subtree $t^2=t_2/q$ of $t_2$, of the form
$t^2=\rho_2[t^2_1,\ldots,t^2_k]$ with $\sigma\leq^{\Delta'_2}\rho_2$,
$s_j\leq_G^{\Gamma'_j}t^2_j$ for $j=1,\ldots,k$, and
$\Delta'_2=\Gamma'_1\cdots\Gamma'_k$.

By the induction hypothesis, there exists flow trees
$s'_1,\ldots,s'_k$ such that
$t^1_j\leq_G^{\Gamma'_j}s'_j$ and
$t^2_j\leq_G^{\Gamma_j}s'_j$ for all $j=1,\ldots,k$.
We now define
\begin{xalignat*}{3}
u&\egdef(\rho_2+\Delta'_1)[s'_1,\ldots,s'_k],
&
u'&\egdef t_2[u]_q,
&
s'&\egdef t_1[u']_p,
\end{xalignat*}
and claim that these are valid flow trees, $s'$ being the flow
tree witnessing the Lemma.

To begin with, and since $\Delta'_1=\Gamma_1\cdots\Gamma_k$, $u$ is
well-formed by Lemma~\ref{lem-replace-children} and satisfies
$t^2\leq_G^{\Delta'_1}u$. Since $\rho_2+\Delta'_1=\rho_1+\Delta'_2$,
and since $\Delta'_2=\Gamma'_1\cdots\Gamma'_k$, one also has
$t^1\leq_G^{\Delta'_2}u$.

Then, and since $t^2\leq_G^{\Delta'_1} u$, we have
$t_2=t_2[t^2]_q\leq_G^{\Delta'_1}t_2[u]_q=u'$ as in
Claim~\ref{cl-replace-one-subtree}.  Thus the root of $u'$ is
$\sigma+\Delta_2+\Delta'_1=\rho_1+\Delta_2$.  We deduce
$t^1\leq_G^{\Delta_2}u'$, relying on $t^1\leq_G u$.  As in
Claim~\ref{cl-replace-one-subtree}, we obtain
$t_1\leq_G^{\Delta_2}t_1[u']_p=s'$, proving the first half of the
Lemma.

On the other hand, from $t_2\leq_G u'$ we get $t_2\leq_G^{\Delta_1}
s'$ by just checking that the root of $t_2$, i.e., $\sigma+\Delta_2$,
is smaller than the root of $s'$, i.e., $\sigma+\Delta_1+\Delta_2$.
This provides the other half and completes the proof.
\end{proof}

\section{GVAS-Definable Predicates}
\label{sec-gvas-sets}

We explore in this section a natural notion of computable sets and relations
for the GVAS model, defined as projections of reachability sets.
The context-free grammar ingredient of GVASes is essential in the
proof that the class of
computable
sets is closed under intersection, while the Amalgamation Theorem proves that computable sets are finite union
of shifted periodic sets.

\begin{definition}
  A $n$-dimensional \emph{GVAS-definable predicate} is a subset $\vec{X}$ of $\Nat^n$
  such that there exists a $d$-dimensional GVAS $G$ with
  $d=n+\ell$ for some $\ell\in\Nat$ such that:
\begin{equation} \tag{$\dagger$}\label{eq-gvas-set}
  \vec{X}=\{ \vx\in\Nat^n \mid \exists \ve\in\Nat^\ell : \vzero_d\xrightarrow{G}(\vx,\ve) \} \:.
\end{equation}
  When \eqref{eq-gvas-set} holds, we say that $G$ defines $\vec{X}$ using $\ell$ auxiliary counters.
\end{definition}

The class of GVAS-definable predicates is clearly closed under union,
cartesian product~\footnote{defined via $\vec{X}\times\vec{Y} \egdef \{ (\vx,\vy) \mid \vx\in\vec{X},\ \vy\in\vec{Y}\}$.}, and by
projecting away some components.
We will provide additional closure properties in the remainder of this section.
GVAS-definable predicates form a rich class that strictly contains all Presburger sets, i.e., 
subsets of $\Nat^n$ that are definable in $\FO(\Nat;+)$, the first-order
theory of natural numbers with addition.
\begin{remark}
  Presburger sets are GVAS-definable predicates. The proof is obtained
  by introducing the class of semilinear sets as follows. 
  A \emph{linear set} of $\Nat^n$ is a set of the form
  $\{\vec{b}+\lambda_1\vec{p}_1+\cdots+\lambda_k\vec{p}_k \mid \lambda_j\in\Nat\}$
  where $\vec{b}$ and $\vec{p}_1,\ldots,\vec{p}_k$ are vectors in
  $\Nat^n$. A \emph{semilinear set} of $\Nat^n$ is a finite union of
  linear sets of $\Nat^n$. Let us recall that a subset of $\Nat^n$
  is Presburger if, and only if, it is semilinear~\cite{gs66}. Since
  the class of GVAS-definable predicates is closed under union, it is sufficient to show that
  every linear set is GVAS-definable. We associate with a linear set
  $\vec{X}=\{\vec{b}+\lambda_1\vec{p}_1+\cdots+\lambda_k\vec{p}_k \mid
  \lambda_j\in\Nat\}$ the
  $n$-dimensional GVAS $G$ that generates the regular language $\vec{b}\vec{p}_1^*\ldots\vec{p}_k^*$.
  Notice that $G$ defines the linear set $\vec{X}$ (using no auxiliary counter).
\end{remark}

\begin{example}
\label{ex-power-2-as-gvas-set}
Let $G$ be the 2-dimensional GVAS given by the following four rules:
\begin{xalignat*}{2}
S & \rightarrow \V{0}{1} ~\big|~ \V{1}{0} \, S \, T       \:,
&
T & \rightarrow \V{0}{0} ~\big|~ \V{0}{-1} \, T \, \V{0}{2} \:.
\end{xalignat*}
This GVAS is a variant of the 1-dimensional GVAS given in
Example~\ref{ex-power-2}.
Analogously to that example,
it can be shown that $G$ defines the set
$\vec{X} = \{ (x, y) \in \Nat^2 \mid 1 \leq y \leq 2^x \}$,
which is not semilinear.
\end{example}

A geometrical decomposition of the GVAS-definable predicates can be shown thanks to the
\emph{periodic sets}. 
A subset $\vec{P}$ of $\Nat^d$ is said to be
\emph{periodic}~\cite{DBLP:conf/popl/Leroux11} if it contains the zero
vector, and if
$\vec{x}+\vec{y}\in\vec{P}$ for every $\vec{x},\vec{y}\in\vec{P}$. 
Note that a periodic $\vec{P}$ is not
necessarily finitely generated, hence not necessarily semilinear (see
Example~\ref{ex-power-2-as-gvas-set-cont}).
The following Proposition extends the  known decomposition of Presburger sets
into linear sets.
\begin{proposition}
\label{prop-gvas-set-periodic}
  Every GVAS-definable predicate $\vec{X}\subseteq \Nat^n$ can be decomposed into a
  finite union of sets of the form $\vec{b}+\vec{P}$ where
  $\vec{b}\in\Nat^n$ and $\vec{P}$ is a periodic subset of $\Nat^n$.
\end{proposition}
\begin{proof}
  There exists a $d$-dimensional GVAS $G$ that defines the set $\vec{X}$ with $\ell$
  auxiliary counters. Let us consider the set $T$ of flow trees $t$
  such that $\root({t})=(\vec{0}_{d},S,(\vec{x},\vec{e}))$ for some
  $\vec{x}\in\Nat^n$, $\vec{e}\in\Nat^\ell$ and where $S$ is the
  start symbol of $G$. For such a flow tree $t$ in $T$, we denote by
  $\mu(t)$ the vector $\vec{x}$.
  With each $s\in T$ we associate
  the set ${\uparrow} s=\{t\in T \mid s \leq_G t\}$. 
  Since $(T,\leq_G)$ is a wqo, there exists a finite
  subset $T_0$ of $T$ such that
  \[ T = \bigcup_{s\in T_0} {\uparrow} s \:. \]
  Given $s\in T$, we introduce the set $\vec{P}_s=\{\mu(t)-\mu(s)
  \mid t\in {\uparrow} s\}$. Theorem~\ref{thm:amalgamation} shows that
  $\vec{P}_s$ is a periodic set. Now, just observe that the following
  equality holds:
  \[ \vec{X}=\bigcup_{s\in T_0}\mu(s)+\vec{P}_{s} \:. \]
  The proposition is proved.
\end{proof}

\begin{example}
\label{ex-power-2-as-gvas-set-cont}
Continuing Example~\ref{ex-power-2-as-gvas-set},
the set
$\vec{X} = \{ (x, y) \in \Nat^2 \mid 1 \leq y \leq 2^x \}$
may be decomposed into
$\vec{X} = (0, 1) + \vec{P}$
where $\vec{P}$ is the periodic set
$\vec{P} = \{ (x, y) \in \Nat^2 \mid 0 \leq y < 2^x \}$.
\end{example}

The rest of this section discusses various closure properties of
GVAS-definable predicates.
We start with boolean operations.
As mentioned previously,
GVAS-definable predicates are closed under union.
In order to prove closure under intersection,
we first provide a technical lemma that shows how auxiliary counters
of a GVAS can be assumed to be zero at the end of the computation.
\begin{lemma}\label{lem:zero}
  For every $d$-dimensional GVAS $G$ and for every subset
  $I$ of $\{1,\ldots,d\}$,
    there is a $(d+1)$-dimensional GVAS $G_I$ such that for every $\vx\in\Nat^d$
  and for every $c\in \Nat$, we have:
\begin{equation}
\label{eq-lem:zero}
\vzero_{d+1}\xrightarrow{G_I}(\vx,c)
\text{ iff } \vzero_d\xrightarrow{G}\vx \land
  c=0 \land \bigwedge_{i\in I}\vx[i]=0
\:.
\end{equation}
\end{lemma}
\begin{proof}
  The idea of the proof is to put the counters in $I$ ``on a budget''
  (see, e.g.,~\cite{phs-mfcs2010} for details on the budgeting
  construction) and to harness
  the expressive power given by context-free grammars to
  non-deterministically initialize the total budget, simulate $G$ with
  the given budget, and finally
   check that the budget is fully restored at the end of the computation,
  which guarantees that the counters in $I$ are zero.

  \medskip

  Formally, let us introduce the function $\Delta_I$ that maps vectors
  $\vx$ of $\setZ^d$ to the number $\Delta_I(\vx)=\sum_{i\in I}\vec{x}[i]$.
  We also introduce the mapping
  $\mu_I:\setZ^d\rightarrow\setZ^{d+1}$ defined by
  $\mu_I(\vx)=(\vx,-\Delta_I(\vx))$. This mapping is
  extended over words of actions as a word morphism, and over
  languages by $\mu_I(L)=\{\mu_I(w) \mid w\in L\}$.
  Let us introduce the actions
  $\vec{a}_+=(\vzero_d,1)$, and
  $\vec{a}_-=(\vzero_d,-1)$.
  In linear time, from $G$ we can define a $(d+1)$-dimensional
  GVAS $G_I$ that generates the following language:
  \[ L_{G_I}=\bigcup_{k\in\Nat}\vec{a}_+^k\mu_I(L_G) \vec{a}_-^k \:. \]
  Let us prove that this GVAS satisfies the lemma.
  We consider $\vx\in\Nat^d$ and $c\in\Nat$. 
  
  \medskip
  
  Assume first that 
  $\vzero_{d+1}\xrightarrow{G_I}(\vx,c)$. In that case, there exists
  $k\in\Nat$ and $w\in L_G$ such
  that $\vzero_{d+1}\xrightarrow{\vec{a}_+^k\mu_I(w) \vec{a}_-^k}(\vx,c)$. Observe
  that we have
  \[
  \vzero_{d+1}\xrightarrow{\vec{a}_+^k}(\vzero_d,k)\xrightarrow{\mu_I(w)}(\vx,c+k)\xrightarrow{\vec{a}_-^k}(\vx,c) \:. \]
  Since $\mu_I(w)$ preserves the sum of the counters in $I$ and of the last counter,
  it follows that
  $\Delta_I(\vzero_d)+k=\Delta_I(\vx)+c+k$. Thus
  $c+\Delta_I(\vx)=0$. It follows that $c=0$ and $\vx[i]=0$ for every
  $i\in I$. Moreover, from $\vzero_d\xrightarrow{w}\vx$ we derive $\vzero_d\xrightarrow{G}\vx$.

  \medskip

  Conversely, let us assume that $\vzero_d\xrightarrow{G}\vx$, $c=0$ and $\vx[i]=0$ for every
  $i\in I$. There exists $w\in L_G$ such that
  $\vzero_d\xrightarrow{w}\vx$. There exists $k\in\Nat$ large enough
  such that
  $(\vzero_d,k)\xrightarrow{\mu_I(w)}(\vx,k-\Delta_I(\vx))=(\vx,k)$. It follows
  that
  $\vzero_{d+1}\xrightarrow{\vec{a}_+^k\mu_I(w)\vec{a}_-^k}(\vx,0)=(\vx,c)$. Thus
  $\vzero_{d+1}\xrightarrow{G_I}(\vx,c)$.
\end{proof}

We are now ready to prove that GVAS-definable predicates are closed
under intersection.\footnote{By contrast, we believe that
``VAS-definable'' predicates are not closed under intersection (unless
one requires  auxiliary counters to be zero at the end of the
computation). This conjecture remains to be proved.}
\begin{lemma}
  \label{lemma:closure-intersection}
  The class of GVAS-definable predicates is closed under intersection.
\end{lemma}
\begin{proof}
  Let $\vec{X},\vec{Y}\subseteq \setN^n$ be GVAS-definable. Since the
  class of GVAS-definable predicates is closed under cartesian product, it follows that
  $\vec{X}\times\vec{Y}$ is also GVAS-definable.
  Hence, there exists a $d$-dimensional GVAS $H$
  with $\ell$ auxiliary counters that defines that set. Let us consider
  the mapping $\mu:\setZ^d\rightarrow \setZ^n\times\setZ^d$ defined by
  $\mu(\vec{x},\vec{y},\vec{e})=(\vec{0}_n,\vec{x},\vec{y},\vec{e})$
  for every $\vec{x},\vec{y}\in\setZ^n$ and
  $\vec{e}\in\setZ^\ell$. The mapping $\mu$ is extended as a word morphism.
  We introduce the
  action $\vec{a}_i$ in $\setZ^{n+d}$ defined as follows:
  \[ \vec{a}_i=(\vunit_{i,n},-\vunit_{i,n},-\vunit_{i,n},\vzero_\ell) \:. \]
    Obviously, we can build
  a $(d+n)$-dimensional GVAS $G$ such
  that $L_{G}=\mu(L_H)\vec{a}_1^*\cdots\vec{a}_d^*$.
  Now, let $I=\{n+1,\ldots,3n\}$ and let us apply Lemma~\ref{lem:zero}
  on $G$ and $I$. We obtain a $(d+n+1)$-GVAS that defines
  $\vec{X}\cap\vec{Y}$.
\end{proof}

\begin{remark}
  The class of GVAS-definable predicates is not closed under taking
  complements, see Proposition~\ref{prop:stabcomplement}.
\end{remark}

We now investigate closure under
sum~\footnote{defined via $\vec{X} + \vec{Y} \egdef \{ \vx + \vy \mid \vx\in\vec{X},\ \vy\in\vec{Y}\}$.}
and under the associated Kleene star, which we call periodic hull.
Formally,
the \emph{periodic hull} of a subset $\vec{X} \subseteq \setN^n$ is the set
of finite sums of vectors in $\vec{X}$.
It turns out that the class of GVAS-definable predicates is closed under sum and periodic hull.
Closure under sum can be proved along the same lines as closure under intersection
(see Lemma~\ref{lemma:closure-intersection}).
The detailed proof is left as an exercise.
Closure under periodic hull is more involved and requires well-behaved GVASes.

\smallskip

In the definition of a GVAS-definable predicate $\vec{X}$ given in~(\ref{eq-gvas-set}), the
vector $\vec{e}$ can be seen as auxiliary counters that
can have arbitrary values
at the end of the computation. We say that a $d$-dimensional GVAS $G$ defining a
predicate $\vec{X}\subseteq \setN^n$ using $\ell$ auxiliary
counters is \emph{auxiliary-resetting} if for every
$(\vec{x},\vec{e})\in\setN^n\times\setN^\ell$ such that
$\vec{0}_d\xrightarrow{G}(\vx,\ve)$ we have $\ve=\vec{0}_\ell$. We
also say that $G$ is \emph{output-increasing} if every action of $G$ is an
element of $\setN^n\times\setZ^\ell$ meaning that the output counters $\vx$
cannot be decremented during a computation.

Let us first prove that GVAS-definable predicates can be defined by auxiliary-resetting output-increasing
GVAS. To do so, we introduce, for every $k\in\setN$, the unit vector $\vec{i}_{i,k}$ of
$\setN^k$ defined by $\vec{i}_{i,k}[j]=0$ if $j\not=i$ and
$\vec{i}_{i,k}[i]=1$.
\begin{lemma}\label{lem:strongdef}
  For every $d$-dimensional GVAS $G$ defining a set $\vec{X}\subseteq\setN^n$,
    there is a $(d+n+1)$-dimensional auxiliary-resetting
  output-increasing GVAS defining $\vec{X}$.
\end{lemma}
\begin{proof}
  We introduce the
  mapping $\mu:\setZ^d\rightarrow\setZ^{n+d}$ defined by
  $\mu(\va)=(\vec{0}_n,\va)$. This mapping is extended
    as a word morphism.
    Observe that there exists a
  $(d+n)$-dimensional GVAS $G'$ satisfying the following equality:
\[
  L_{G'}=\mu(L_G)(\vec{i}_{1,n},-\vec{i}_{1,n},\vec{0}_\ell)^*\cdots
  (\vec{i}_{n,n},-\vec{i}_{n,n},\vec{0}_\ell)^*(\vec{0}_n,\vec{0}_n,-\vec{i}_{1,\ell})^*\cdots
  (\vec{0}_n,\vec{0}_n,-\vec{i}_{\ell,\ell})^* \:.
\]
  This GVAS $G'$ satisfies
  $\vec{0}_{d+n}\xrightarrow{G'}(\vx,\vec{0}_n,\vec{0}_\ell)$ with $\vx\in
  \setN^n$ if, and only if, there exists
  $\ve\in\setN^\ell$ such that 
  $\vec{0}_d\xrightarrow{G}(\vx,\ve)$. Moreover, actions of that GVAS
  are in $\setN^n\times\setZ^d$.
  By applying the construction given in the proof of Lemma \ref{lem:zero} on
  $G'$ with $I=\{n,\ldots,d+n\}$ observe that we get a $(d+n+1)$-dimensional auxiliary-resetting
  output-increasing GVAS defining $\vec{X}$. 
\end{proof}

\begin{corollary}
  The class of GVAS-definable predicates is closed under periodic hull.
\end{corollary}
\begin{proof}
  Assume that a predicate $\vec{X}\subseteq \setN^n$ is
  GVAS-definable. Lemma~\ref{lem:strongdef} shows that $\vec{X}$ is
  defined by a $d$-dimensional auxiliary-resetting
  output-increasing GVAS $G$ using $\ell$ auxiliary counters.
  As context-free languages are closed under Kleene star,
  there exists a GVAS $G'$ satisfying
  $L_{G'} = L_G^*$.
  We show that $G'$ defines the periodic hull of $\vec{X}$. Let
  $\vec{y}=\vec{x}_1+\cdots+\vec{x}_k$ with $k\in\setN$ and
  $\vec{x}_1,\ldots,\vec{x}_k\in\vec{X}$ and let us prove that
  $\vec{0}_d\xrightarrow{G'}(\vec{y},\vec{0}_\ell)$. We introduce
  $\vec{y}_j=\vec{x}_1+\cdots+\vec{x}_j$ for every
  $j\in\{0,\ldots,k\}$. Let $i\in\{1,\ldots,k\}$. Since $\vx_i\in
  \vec{X}$ and $G$ is auxiliary-resetting,
  we get $\vec{0}_d\xrightarrow{G}(\vec{x}_i,\vec{0}_\ell)$. By
  monotony, we can add on both sides the vector
  $(\vec{y}_{i-1},\vec{0}_\ell)$ and derive
  $(\vec{y}_{i-1},\vec{0}_{\ell})\xrightarrow{G}(\vec{y}_{i},\vec{0}_\ell)$.
  We get $(\vec{y}_0,\vec{0}_{\ell}) \xrightarrow{G} (\vec{y}_{1},\vec{0}_\ell) \cdots \xrightarrow{G} (\vec{y}_k,\vec{0}_\ell)$.
  Since $\vec{y}_0=\vec{0}_n$ and $\vec{y}_k=\vec{y}$, we have proved that
  $\vec{0}_d\xrightarrow{G'}(\vec{y},\vec{0}_\ell)$. Conversely, let
  $\vec{y}\in\setN^n$ and $\vec{e}\in\setN^\ell$ such that
  $\vec{0}_d\xrightarrow{G'}(\vec{y},\vec{e})$ and let us prove that
  $\vec{y}$ is in the periodic hull of $\vec{X}$.
  Since $L_{G'} = L_G^*$, we have
  $(\vec{y}_0,\vec{e}_0) \xrightarrow{G} (\vec{y}_{1},\vec{e}_1) \cdots \xrightarrow{G} (\vec{y}_k,\vec{e}_k)$
  for some sequence $(\vec{y}_0,\vec{e}_0), \ldots, (\vec{y}_k,\vec{e}_k)$ such that
  $(\vec{y}_0,\vec{e}_0)=\vec{0}_d$ and $(\vec{y}_k,\vec{e}_k)=(\vec{y},\vec{e})$.
  Since
  $G$ is output-increasing we deduce that $\vec{x}_i$, defined as
  $\vec{x}_i = \vec{y}_i-\vec{y}_{i-1}$, is in $\setN^n$ and satisfies 
  $(\vec{0}_n,\vec{e}_{i-1})\xrightarrow{G}(\vec{x}_i,\vec{e}_i)$.
  As $G$ is auxiliary-resetting, by induction we deduce that
  $\vec{e}_i=\vec{0}$ for every $i$. It follows that
  $\vec{x}_i\in\vec{X}$. As
  $\vec{y}=\vec{x}_1+\cdots+\vec{x}_k$, we conclude that
  $\vec{y}$ is in the periodic hull of $\vec{X}$.
\end{proof}

To conclude this section,
we discuss closure under relational composition
and under the associated Kleene star (namely, the reflexive-transitive closure).
For the purpose of GVAS-definability,
we view binary relations on $\Nat^n$ as subsets of $\Nat^{2n}$.
The class of GVAS-definable binary relations on $\Nat^n$ is closed under
relational composition.
This claim follows from closure of GVAS-definable predicates under
cartesian product, intersection and projection.
However,
GVAS-definable binary relations are not closed under reflexive-transitive closure in general,
as the following example shows.

\begin{example}
  Consider the binary relation $\vec{R}$ on $\Nat$ defined by $\vec{R} = \{(x, 2x) \mid x \in \Nat\}$.
  The binary relation $\vec{R}$ is clearly GVAS-definable.
    However, its reflexive-transitive closure
  $\vec{R}^* = \{(x, 2^k x) \mid x, k \in \Nat\}$
  is not GVAS-definable.
  Indeed,
  if $\vec{R}^*$ were GVAS-definable then,
  by Proposition~\ref{prop-gvas-set-periodic},
  there would exist some $\vec{b} \in \Nat^2$, 
  some periodic $\vec{P} \subseteq \Nat^2$ 
  and two distinct powers $2\leq 2^k<2^\ell$  such that
  $$
  (1, 2^k), (1, 2^\ell) \in \vec{b} + \vec{P} \subseteq \vec{R}^*  \:.
  $$
  We now use the assumption that $\vec{P}$ is periodic and derive a contradiction.
  Let us write $\vec{b} = (b_1, b_2)$.
  Note that $b_1 \leq 1$ and $b_2 \leq 2^k$ since $(1, 2^k) \in \vec{b} + \vec{P}$.
  There are three cases.
  \begin{itemize}
  \item
    If $\vec{b} = \vec{0}$ then $(1, 2^k)$ and $(1, 2^\ell)$ are both in $\vec{P}$,
    hence,
    $(2, 2^k + 2^\ell)$ is also in $\vec{P}$ by periodicity,
    and so $(2, 2^k + 2^\ell) \in \vec{R}^*$.
    This is impossible since $2^k + 2^\ell$ is not a power of two (as $k \neq \ell$).
  \item
    If $b_1 > 0$ then $b_1 = 1$, hence, $(1, 2^\ell) = (1, b_2) + (0, p_2)$
    for some $(0, p_2) \in \vec{P}$.
    Note that $p_2 > 0$ since $b_2 \leq 2^k < 2^\ell$.
    Since $\vec{b} + \vec{P} \subseteq \vec{R}^*$,
    we get by periodicity that $(1, b_2 + n p_2) \in \vec{R}^*$
    for every $n \in \Nat$.
    This means that $b_2 + n p_2$ is a power of two for every $n \in \Nat$,
    which is impossible as $p_2 > 0$.
  \item
    If $b_1 = 0$ and $b_2 > 0$ then $(1, 2^\ell) = (0, b_2) + (1, p_2)$
    for some $(1, p_2) \in \vec{P}$.
    Since $\vec{b} + \vec{P} \subseteq \vec{R}^*$,
    we get by periodicity that $(n, b_2 + n p_2) \in \vec{R}^*$
    for every $n \in \Nat$.
    Taking $n = b_2 + 1$,
    we derive that $(b_2 + 1, b_2 + (b_2 + 1) p_2) \in \vec{R}^*$.
    This is impossible since $b_2 > 0$ and $y$ is a multiple of $x$
    for every $(x, y) \in \vec{R}^*$.
  \end{itemize}
\end{example}

 \section{Weakly Computable Functions}
\label{sec-weakcomp}

There is a classical notion of number-theoretical functions
weakly computable by Petri nets~\cite{hack76b}. In this section, we extend
the idea to GVASes.

As we argued in the introduction, the notion of weakly computable
functions has recently gained new relevance with the development of
well-structured systems that go beyond Petri nets and VASSes in
expressive power, while sharing some of their characteristics.

\bigskip

The expected way for a GVAS to compute a numerical function $f:\Nat\to\Nat$ is to start with some
input number $n$ stored in a designated input counter and, 
from that configuration, eventually reach a configurations with $f(n)$ in a designated output counter. In
order for that GVAS to be correct (as a computer for $f$), it should be impossible that it
reaches a value differing from $f(n)$ in
the output counter. In that case, we say that the GVAS \emph{strongly
  computes} $f$. This notion of correctness is fine with other models
like Minsky
machines but it is too strong for GVASes and does not lead to an
interesting family of computable functions. In fact, GVAS
are essentially nondeterministic devices, and the above notion
of strongly computing some function does not accommodate nondeterminism
nicely.

With this in mind, and in the setting of VASes, Rabin defined a notion of ``weakly computing $f$'' that
combines the following two principles:
\begin{description}
\item[Completeness]
For any $n\in\Nat$,
there is a computation with input $n$ and output $f(n)$;
\item[Safety]
Any computation from input $n$ to some output $r$ satisfies $r\leq f(n)$.
\end{description}
This leads to our definition of weak GVAS computers,
where the input and output counters are the first two components.
\begin{definition}[Weak GVAS computers]
\label{def-wpn}
Let $f:\Nat\to\Nat$ be a total function. A \emph{weak GVAS computer}
(with $\ell$ auxiliary counters) for $f$ is a $d$-dimensional GVAS $G$  with
 $d=2+\ell$ that satisfies the
following two properties:
\begin{align}
  \tag{CO}\label{eq-co}
    \forall n: \exists n',\vec{e}:\: &
                                (n,0,\vzero_\ell)
                                \xrightarrow{G}
                                (n',f(n),\vec{e})
                                \:,
  \\
  \tag{SA}\label{eq-sa}
    \forall n, n',r,\vec{e}:\: &
                            (n,0,\vzero_\ell)\xrightarrow{G}(n',r,\vec{e})
                            \text{ implies } r\leq f(n)
                            \:.
\end{align}
We say that $f$ is \emph{weakly computable}, or WC, if there is a weak
GVAS computer for it.
\end{definition}
For convenience, Definition~\ref{def-wpn} assumes that the
input is given in the first counter of $G$, and that the
result is found in the second counter. Note that
$G$ may use its $\ell$ last counters for auxiliary
calculations. We focus on total functions over the natural numbers
rather than total functions over the vectors of natural numbers to simplify the
presentation. However, results given in this section can be easily
extended to this more general setting.

\begin{example}[A weak computer for exponentiation]
  \label{ex:power-WC}
  Example~\ref{ex-power-2} shows that the function
  $f:\Nat\rightarrow\Nat$ defined by $f(n)=2^{n}$ is WC.
  \qed
\end{example}

Only monotonic functions can be weakly computed in the above sense.
This is an immediate consequence of the monotonicity of
GVASes (see \eqref{eq-gvas-monotonic}).
Recall that a total function $f:\Nat\to\Nat$ is \emph{non-decreasing} if
$n\leq m$ implies $f(n)\leq f(m)$.
\begin{proposition}[Monotonicity of WC functions]
\label{prop-WCPN-mono}
If $f$ is WC then $f$ is non-decreasing.
\end{proposition}
\begin{proof}
Assume that $n\leq m$ and pick any weak GVAS computer $G$ for $f$.
By \eqref{eq-co}, we have
$(n,0,\vzero_\ell) \xrightarrow{G}
(n',f(n),\vec{e})$ for some $n'\in\Nat$ and $\vec{e}\in\Nat^\ell$.
By monotonicity, it follows that $(n+(m-n),0,\vzero_\ell) \xrightarrow{G}
(n'+(m-n),f(n),\vec{e})$. We get $f(n)\leq f(m)$ by \eqref{eq-sa}.
\end{proof}

We may now relate WC computability with
GVAS-definability.
\begin{lemma}\label{lem:WCGVAS}
  A total function $f:\Nat\to\Nat$ is WC if, and only if,
  $f$ is non-decreasing and the following set is GVAS-definable.
  \[ \{(x,y)\in\setN\times\setN \mid y\leq f(x)\} \:. \]
\end{lemma}
\begin{proof}
  Assume first that $f$ is WC. There exists a weak GVAS computer
  (with $\ell$ auxiliary counters) for $f$ given as a $d$-dimensional GVAS
  $G$  with $d=2+\ell$ that satisfies (\ref{eq-co}) and
  (\ref{eq-sa}). Let us consider the mapping
  $\mu:\setZ^d\rightarrow\setZ^{d+1}$ defined by
  $\mu(a,b,\vec{e})=(0,b, a,\vec{e})$ for every $a,b\in\setZ$, and
  $\vec{e}\in\setZ^\ell$.
  The mapping $\mu$ is extended as a word
  morphism. Let us show that a GVAS $G'$ such that
  $L_{G'}=(1,0,1,\vec{0}_\ell)^*\mu(L_G)(0,-1,0,\vec{0}_\ell)^*$ is
  defining the set $\{(x,y)\in\setN\times\setN \mid y\leq f(x)\}$. Let
  $(x,y)$ in that set. From (\ref{eq-co}), there exists a word
  $\sigma\in L_G$, $x'\in\setN$ and $\vec{e}\in\setN^\ell$ such that
  $(x,0,\vec{0}_\ell)\xrightarrow{\sigma}(x',f(x),\vec{e})$. The word
  $\sigma'=(1,0,1,\vec{0}_\ell)^x\mu(\sigma)(0,-1,0,\vec{0})^{f(x)-y}$
  shows that
  $(0,0,0,\vec{0}_\ell)\xrightarrow{\sigma'}(x,y,x',\vec{e})$. As
  $\sigma'\in L_{G'}$, we get
  $(0,0,0,\vec{0}_\ell)\xrightarrow{G'}(x,y,x',\vec{e})$. Conversely,
  assume that $(0,0,0,\vec{0}_\ell)\xrightarrow{\sigma'}(x,y,x',\vec{e})$
  for some $x,y,x'\in\setN$, $\vec{e}\in\setN^\ell$ and $\sigma'\in
  L_{G'}$, and let us prove
  that $y\leq f(x)$. By definition of $G'$, there exists $n,m\in\setN$
  and a word $\sigma\in L_G$ such that
  $\sigma'=(1,0,1,\vec{0}_\ell)^n\mu(\sigma)(0,-1,0,\vec{0}_\ell)^m$. It
  follows that
  $(n,0,n,\vec{0}_\ell)\xrightarrow{\mu(\sigma)}(x,y+m,x',\vec{e})$. Since
  actions occurring in $\mu(\sigma)$ cannot modify the first counter,
  we get $n=x$. Moreover,
  $(x,0,\vec{0}_\ell)\xrightarrow{\sigma}(x',y+m,\vec{e})$. From
  (\ref{eq-sa}), we derive $y+m\leq f(x)$. Hence $y\leq f(x)$. We have
  proved that $\{(x,y)\in\setN\times\setN \mid y\leq f(x)\}$ is
  GVAS-definable.

  \medskip

  Conversely, let us assume that $f$ is non-decreasing and that $\{(x,y)\in\setN\times\setN \mid
  y\leq f(x)\}$ is GVAS-definable. There exists a $d$-dimensional GVAS $G$ with
  $d=2+\ell$ such that:
  \[ \{(x,y) \mid y\leq f(x)\}=\{(x,y)\mid \exists
  \vec{e}\in\setN^\ell : (0,0,\vec{0}_\ell)\xrightarrow{G}(x,y,\vec{e})\} \:. \]
  Let us consider the mapping
  $\mu:\setZ^d\rightarrow\setZ^{d+1}$ defined by
  $\mu(a,b,\vec{e})=(-a,b, a,\vec{e})$ for every $a,b\in\setZ$, and $\vec{e}\in\setZ^\ell$.
  The mapping $\mu$ is extended as a word
  morphism. Let us show that a GVAS $G'$ such that
  $L_{G'}=\bigcup_{k\in\setN}(1,0,0,\vec{0}_\ell)^k\mu(L_G)(-1,0,0,\vec{0}_\ell)^k$
  is a weak GVAS computer for $f$. Let us first consider
  $x\in\setN$. By definition of $G$, there exists a word $\sigma\in
  L_G$ and $\vec{e}\in\setN^\ell$ such that
  $(0,0,\vec{0}_\ell)\xrightarrow{\sigma}(x,f(x),\vec{e})$. Notice
  that for $k$ large enough, we have $(k+x,0,0,\vec{0}_\ell)
  \xrightarrow{\mu(\sigma)}(k,f(x),x,\vec{e})$. The word
  $\sigma'= (1,0,0,\vec{0}_\ell)^k\mu(\sigma)(-1,0,0,\vec{0}_\ell)^k$
  is such that
  $(x,0,0,\vec{0}_\ell)\xrightarrow{\sigma'}(0,f(x),x,\vec{e})$. Hence
  (\ref{eq-co}) is satisfied by $G'$. Finally, let us assume that
  $(x,0,0,\vec{0}_\ell)\xrightarrow{\sigma'}(z,y,x',\vec{e})$ for a
  word $\sigma'\in L_{G'}$ and $x',y,z\in\setN$ and
  $\vec{e}\in\setN^\ell$. There exists $k\in\setN$ and $\sigma\in L_G$
  such that
  $\sigma'=(1,0,0,\vec{0}_\ell)^k\mu(\sigma)(-1,0,0,\vec{0}_\ell)^k$. It
  follows that
  $(x+k,0,0,\vec{0}_\ell)\xrightarrow{\mu(\sigma)}(z+k,y,x',\vec{e})$. By
  definition of $\mu$, since the effect of the sum of the first and third
  counters is zero, we get $x+k+0=z+k+x'$. Hence $x'\leq x$ and in
  particular $f(x')\leq f(x)$. Moreover, we
  have $(0,0,\vec{0}_\ell)\xrightarrow{\sigma}(x',y,\vec{e})$. By
  definition of $G$, we get $y\leq f(x')$. We have proved that $y\leq
  f(x)$. Hence (\ref{eq-sa}) is satisfied by $G'$. We have proved that
  $G'$ is a weak GVAS computer for $f$
\end{proof}

By combining Lemma~\ref{lem:WCGVAS} and the decomposition of
GVAS-definable sets given by Proposition~\ref{prop-gvas-set-periodic},
we obtain two interesting, albeit negative, results on WC functions
and GVAS-definable sets.
\begin{proposition}\label{prop:sublin}
Let $f$ be an unbounded WC function. Then there exists a rational
number $c>0$ and some $z\in\setZ$ such that $f(n)\geq c n+z$ for every $n\in\Nat$. 
\end{proposition}
\begin{proof}
  Lemma~\ref{lem:WCGVAS} shows that the set $\vec{X}$ defined as $\{(n,m) \mid m\leq f(n)\}$ is
  GVAS-definable. Proposition~\ref{prop-gvas-set-periodic} shows that $\vec{X}$
  can be decomposed into a
  finite union of sets of the form $(a,b)+\vec{P}$ where
  $(a,b)\in\Nat^2$ and $\vec{P}$ is a periodic subset of
  $\Nat^2$. Since $f$ is unbounded, there exists $(p,q)\in \vec{P}$
  such that $q>0$. It follows $(a,b)+k(p,q)\in\vec{X}$ for every
  $k\in\Nat$. In particular $f(a+kp)\geq b+kq$ for every $k\in\Nat$.
  As $f(a)\in\Nat$ and $q>0$, we deduce that $p>0$. Let us consider
  $n\in\setN$ such that $n\geq a$ and observe that there exists
  $k\in\setN$ such that:
  \[ k\leq \frac{n-a}{p}<k+1 \:. \]
  It follows that $a+kp\leq n$ and in particular $f(a+kp)\leq f(n)$.
  Hence $f(n)\geq b+kq\geq b+(\frac{n-a}{p}-1)q$. Introducing
  $c=\frac{q}{p}$, we deduce that $f(n)- c n\geq b-c(a+p)$ for every
  $n\geq a$. We have proved the lemma with any $z\in\setZ$ satisfying
  $z\leq f(n)-cn$ for every $0\leq n<a$ and $z\leq  b-c(a+p)$.
\end{proof}

\begin{proposition}
\label{prop:stabcomplement}
The complement of a GVAS-definable set is not always GVAS-definable.
\end{proposition}
\begin{proof}
  Recall from Example~\ref{ex:power-WC} that the function $f:\Nat\to\Nat$
  defined by $f(n)=2^{n}$ is WC. We derive from Lemma~\ref{lem:WCGVAS}
  that $\vec{X}\egdef\{(n,m) \mid m\leq 2^n\}$ is
  GVAS-definable. Assume, by way of contradiction, that the complement
  $\vec{Y}\egdef\{(n,m) \mid 2^n< m\}$ is GVAS-definable.
  From a GVAS defining $\vec{Y}$,
  we easily derive a GVAS defining $\vec{Z}\egdef\{(n,m) \mid 2^m \leq n+1\}$,
  by swapping the first two counters and then decrementing the first counter
  by two at the end.
  It follows from Lemma~\ref{lem:WCGVAS} that the mapping $g:\Nat\to\Nat$
  defined by $g(n)=\lfloor \log_2 (n+1)\rfloor$ is WC,
   contradicting Proposition~\ref{prop:sublin} since $g$ is unbounded and sublinear.
  Hence $\vec{Y}$, i.e.,
  $\setN^2\setminus\vec{X}$, cannot be GVAS-definable.
              \end{proof}

 \section{Hyper-Ackermannian GVAS}
\label{sec-hypack}

In this section we construct GVASes that weakly compute functions from
the Fast Growing Hierarchy. 
Our main result is the following.
\begin{theorem}
\label{thm-Fa-WC}
The Fast Growing functions $(F_\alpha)_{\alpha<\omega^\omega}$ are
weakly computable (by GVASes).
\end{theorem}
Note that these are exactly the multiply-recursive functions
$F_\alpha$. They include functions that are not primitive-recursive
(the $F_\alpha$ for $\omega\leq\alpha<\omega^\omega$) and that are
thus not weakly computable by VASSes
(see~\cite[section~2]{jantzen80}).  We do not know whether
$F_{\omega^\omega}$ is weakly computable by a GVAS, or whether there
exist WC functions that are not multiply-recursive.
\\

The rest of this section proves Theorem~\ref{thm-Fa-WC}.  The detailed
proof illustrates how the GVAS model makes it manageable to define
complex constructions precisely, and to formally prove their
correctness. By contrast, observe how in less abstract models e.g.,
the Timed-Arc Petri Nets of~\cite{HSS-lics2012}, only schematic
constructions are given for
weakly computing functions, and only an outline for a
correctness proof can be provided.

We follow
notation and definitions from~\cite{schmitz-toct2016} and consider
functions $F_\alpha:\Nat\to\Nat$ indexed by an ordinal
$\alpha<\epsilon_0$ (though we shall only build GVASes for functions
with $\alpha<\omega^\omega$).  Any such ordinal can be written in
Cantor normal form (CNF) $\alpha =
\omega^{\alpha_1}+\cdots+\omega^{\alpha_m}$ with $\alpha>\alpha_1\geq
\cdots\geq\alpha_m$. When $m=0$, $\alpha$ is $0$. When $\alpha_m=0$,
$\alpha$ is a successor of the form $\beta+\omega^0$, i.e., $\beta+1$,
and when $\alpha_m>0$, $\alpha$ is a limit ordinal.  When $\alpha\neq
0$, we often decompose $\alpha$ under the form $\alpha =
\gamma+\omega^{\alpha_m}$ so that the smallest summand in $\alpha$'s
CNF is exposed.  CNFs are often written more concisely using
coefficients, as in $\alpha=\omega^{\alpha_1}\cdot c_1 + \cdots +
\omega^{\alpha_m}\cdot c_m$, with now $\alpha>\alpha_1>
\cdots>\alpha_m$ and $\omega>c_1,\ldots,c_m>0$.

With each limit ordinal $\lambda<\epsilon_0$, one associates a fundamental
sequence $(\lambda(n))_{n<\omega}$  such that
$\lambda=\sup_n\lambda(n)$.  These are defined inductively as follows.
\begin{align}
\label{eq-L1}\tag{L1}
  (\gamma + \omega^{\beta + 1})(n) &= \gamma +
  \omega^{\beta}\cdot(n+1)
\:,
\\
\label{eq-LL}\tag{LL}
 (\gamma + \omega^{\lambda})(n) &= \gamma
  + \omega^{\lambda(n)}
\:.
\end{align}

For instance, Eq.~\eqref{eq-L1} gives $\omega(n)$, i.e., $\omega^1(n)
= \omega^0\cdot(n + 1)=n+1$ and $(\omega^{3}\cdot 6 + \omega^{2}\cdot
3)(n) = \omega^{3}\cdot 6 + \omega^{2}\cdot 2 + \omega\cdot(n+1)$.
Similarly, Eq.~\eqref{eq-LL} gives $\omega^{\omega}(n) = \omega^{\omega(n)}
= \omega^{n+1}$.
Note that the fundamental sequences satisfy $\lambda(0) < \cdots <
\lambda(n) < \lambda(n+1) < \cdots < \lambda$ for any limit ordinal
$\lambda$ and index $n$.

We may now define our fast growing functions $F_\alpha:\Nat\to\Nat$
for $\alpha<\epsilon_0$ by
induction on the $\alpha$ index.
\begin{align}
\label{eq-F0}\tag{F0}
  F_0(x) &= x+1
\:,\\[-1.5em]
\label{eq-F1}\tag{F1}
  F_{\alpha + 1}(x) &= F_{\alpha}^{\omega(x)}(x) =
  \overbrace{F_{\alpha}(\cdots (F_{\alpha}}^{x+1\text{
  times}}(x) )\cdots )
\:,\\
\label{eq-FL}\tag{FL}
  F_{\lambda}(x) &= F_{\lambda(x)}(x)
\:.
\end{align}
As shown ---e.g., in~\cite{schmitz-toct2016}--- these functions are \emph{strictly expansive} and
\emph{monotonic}, i.e., for all ordinals $\alpha<\epsilon_0$ and all $n,n'\in\Nat$:
\begin{gather}
\label{eq-F-exp}\tag{FX}
n < F_\alpha(n)
\:,
\\
\label{eq-F-mono}\tag{FM}
n\leq n' \implies F_\alpha (n)\leq F_\alpha (n')
\:.
\end{gather}

Given two ordinals in Cantor normal form $\alpha=\omega^{\beta_1}+ \cdots
+ \omega^{\beta_m}$ and $\alpha'=\omega^{\beta'_1} + \cdots + \omega^{\beta'_n}$,
we denote by $\alpha \oplus \alpha'$ their \emph{natural sum}
$\sum_{k=1}^{m+n} \omega^{\gamma_k}$, where $\gamma_1 \geq \ldots \geq \gamma_{m+n}$
is a reordering of $\beta_1, \ldots, \beta_m, \beta'_1, \ldots, \beta'_n$.
The $F_\alpha$ functions are not monotonic in the ordinal
index, i.e., $\alpha\leq\alpha'$ does not always entail
$F_\alpha(n)\leq F_{\alpha'}(n)$, see~\cite[section A.2]{schmitz-toct2016}. However, our construction
relies on similar monotonicity properties, albeit for special cases
of $\alpha$ and $\alpha'$, that we now state.
\begin{lemma}
\label{lem-Fmono-alpha}
For any ordinals $\alpha,\alpha'<\epsilon_0$ and any $n\in\Nat$,
$F_{\alpha}(n) \leq F_{\alpha\oplus\alpha'}(n)$.
\end{lemma}
\begin{lemma}
\label{lem-Fmono-alphaSeq}
For any ordinal $\alpha<\epsilon_0$ and limit ordinal $\lambda<\omega^\omega$,
for any $m, n\in\Nat$, if $m \leq n$ then
$F_{\alpha\oplus\lambda(m)}(n)
\leq
F_{\alpha\oplus\lambda}(n)$.
\end{lemma}
For these two results, detailed proofs are given in the appendix.  We
note that Lemma~\ref{lem-Fmono-alpha} is a rewording of Lemma~2.2a
from~\cite{CS-lics08}, however that paper uses a different definition
for the fundamental sequences $(\lambda(n))_{n\in\Nat}$, resulting in
slightly different $F_\alpha$ functions, hence the need of an
independent proof.  Similarly, Lemma~\ref{lem-Fmono-alphaSeq} is a
generalization of Lemma~VI.5 from~\cite{leroux2014}, using different
notation and allowing a simpler proof.  \\

We now define weak GVAS computers for the $F_{\alpha}$ functions such
that $\alpha<\omega^\omega$. Our construction is in two steps: we
first pick an arbitrary exponent $d\in\Nat$ and define $G_d$, a GVAS
with a structure suitable for correctness proofs. We then obtain a
weak GVAS computer for $F_\alpha$ by slightly modifying $G_d$,
provided $\alpha<\omega^d$. The whole construction is an adaptation
into the GVAS framework of the pushdown VAS from~\cite{leroux2014}.

The dimension of $G_d$ is $d+2$ and we use $d+2$
counters named $r, \overline{r}, \kappa_0, \ldots, \kappa_{d-1}$,
in this order.
The set of actions $\vec{A}\subseteq \setZ^{d+2}$ consists     of all
vectors
 $\vec{d}_{x}$ and $\vunit_{x}$ where $x$ is one of the $d+2$ counters:
formally $\vec{d}_x$ is the vector that
decrements $x$, while $\vunit_x\egdef-\vec{d}_x$ increments it.
For instance,
$\vec{d}_{\kappa_0} = (0, 0, -1, \vzero_{d-1})$ and
$\vunit_{r} = (1, 0, \vzero_{d})$.
The set of non-terminals of $G_d$ is
$V=\{\NTF, \REC, \REST, \LIM_1, \ldots, \LIM_{d-1}\}$.  The start
symbol is $\NTF$.  The other non-terminals are used for intermediate
steps (see the rules below).

The first two counters, $r$ and $\overline{r}$, are used to manipulate
the arguments of the functions being computed.  The other $d$ counters
are used as a data structure representing an ordinal
$\alpha<\omega^d$. Formally, with any $d$-tuple
$\tuple{c_0,\ldots,c_{d-1}}$ of natural numbers, we associate the
ordinal $\alpha = \omega^{d-1} \cdot c_{d-1} + \cdots + \omega^{0}
\cdot c_0$.  We will follow the convention of writing the contents of
the counters of our GVAS in the form $\tuple{n, m, \alpha}$, where $n$
and $m$ are the value of $r$ and $\overline{r}$, respectively, and
where $\alpha$ is the ordinal associated with the values in
$\kappa_0,\ldots,\kappa_{d-1}$.

The rules of $G_d$ are given below. The rules involving the
$\LIM_{i}$ non-terminals are present for every $i \in \{1,
\ldots, d-1\}$.
\begin{xalignat}{2}
\NTF   & \rightarrow \vunit_{r} \:, \label{ruleF0}\tag{R1}
\\
\NTF   & \rightarrow \vec{d}_{\kappa_0} \ \REC \ \NTF \  \vunit_{\kappa_0} \:, \label{ruleF+1}\tag{R2}
\\
\NTF   & \rightarrow  \vec{d}_{\kappa_i} \vunit_{\kappa_{i-1}} \ \LIM_i \ \vec{d}_{\kappa_{i-1}}\vunit_{\kappa_{i}} \:, \label{ruleFlim}\tag{R3$i$}
\\
\REC   &\rightarrow \REST \:, \label{ruleRec0}\tag{R4}
\\
\REC   & \rightarrow  \vec{d}_r\vunit_{\overline{r}} \ \REC \ \NTF \:, 
\label{ruleRec1}\tag{R5}
\\
\REST  & \rightarrow \varepsilon \:, \label{ruleRest0}\tag{R6}
\\
\REST  & \rightarrow \vunit_r \vec{d}_{\overline{r}} \ \REST \:, 
\label{ruleRest1}\tag{R7}
\\
\LIM_i & \rightarrow \REST \ \NTF  \:, \label{ruleLim0}\tag{R8$i$}
\\
\LIM_i & \rightarrow  \vec{d}_r\vunit_{\overline{r}} \vunit_{\kappa_{i-1}} \ \LIM_{i} \ \vec{d}_{\kappa_{i-1}} \:. 
\label{ruleLim1}\tag{R9$i$}
\end{xalignat}

Our first goal is to prove that $G_d$ has computations of the form
$\tuple{n,0,\alpha} \xrightarrow{\NTF} \tuple{F_{\alpha}(n),0,\alpha}$,
for any $n\in\Nat$ and $\alpha<\omega^d$.
We start with a lemma exposing some specific sentential forms that can
be derived from $\NTF$ and $\REST$.
As will be clear from the proof of Lemma~\ref{lem:completenessF},
these derivations
(namely \eqref{derivF0}, \eqref{derivF+1} and \eqref{derivFlim})
correspond to our inductive definition of the fast growing functions $F_\alpha$
(namely \eqref{eq-F0}, \eqref{eq-F1} and \eqref{eq-FL}).

\begin{lemma}
  \label{lem:completeness-derivations}
  For every $n \in \Nat$ and $0<i<d$, $G_d$ admits the following derivations:
  \begin{align}
    \label{derivF0}\tag{D0}
    \NTF & \stepstar \vunit_{r}
    \:,
    \\
    \label{derivF+1}\tag{D1}
    \NTF & \stepstar \vec{d}_{\kappa_0} \ (\vec{d}_r \vunit_{\overline{r}})^n \ \REST \ (\NTF)^{n+1} \  \vunit_{\kappa_0}
    \:,
    \\
    \label{derivFlim}\tag{DL}
    \NTF & \stepstar \vec{d}_{\kappa_i} \vunit_{\kappa_{i-1}} \ (\vec{d}_r \vunit_{\overline{r}} \vunit_{\kappa_{i-1}})^n \ \REST \ \NTF \ (\vec{d}_{\kappa_{i-1}})^{n+1} \ \vunit_{\kappa_{i}}
    \:,
    \\
    \label{derivPop}\tag{DP}
    \REST & \stepstar (\vunit_r \vec{d}_{\overline{r}})^n
    \:.
  \end{align}
\end{lemma}
\begin{proof}
  Derivation \eqref{derivF0} is an immediate consequence of rule
  \eqref{ruleF0} and
  derivation \eqref{derivPop} similarly follows from rules \eqref{ruleRest0} and \eqref{ruleRest1}.
  To prove derivation \eqref{derivF+1},
  we use
\begin{align*}
\NTF \xLongrightarrow{\eqref{ruleF+1}} 
\vec{d}_{\kappa_0} \ \REC \ \NTF \ \vunit_{\kappa_0} 
\xLongrightarrow{\eqref{ruleRec1}} \cdots 
& \xLongrightarrow{\eqref{ruleRec1}}
\vec{d}_{\kappa_0} \ (\vec{d}_r\vunit_{\overline{r}})^n \ \REC
\ (\NTF)^n \ \NTF \ \vunit_{\kappa_0} 
\\
& \xLongrightarrow{\eqref{ruleRec0}}
\vec{d}_{\kappa_0} \ (\vec{d}_r\vunit_{\overline{r}})^n \ \REST
\ (\NTF)^{n+1} \ \vunit_{\kappa_0} 
\:.
\end{align*}
Finally, derivation \eqref{derivFlim} is obtained with
\begin{align*}
\NTF
&
\xLongrightarrow{\eqref{ruleFlim}}
\vec{d}_{\kappa_i} \vunit_{\kappa_{i-1}} \ \LIM_i
\ \vec{d}_{\kappa_{i-1}}\vunit_{\kappa_{i}}
\\
&
\xLongrightarrow{\eqref{ruleLim1}}\cdots\xLongrightarrow{\eqref{ruleLim1}}
\vec{d}_{\kappa_i} \vunit_{\kappa_{i-1}} \ 
(\vec{d}_r\vunit_{\overline{r}} \vunit_{\kappa_{i-1}})^n \ \LIM_i \ 
(\vec{d}_{\kappa_{i-1}})^n \ 
\vec{d}_{\kappa_{i-1}}\vunit_{\kappa_{i}}
\\
&
\xLongrightarrow{\eqref{ruleLim0}}
\vec{d}_{\kappa_i} \vunit_{\kappa_{i-1}} \ 
(\vec{d}_r\vunit_{\overline{r}} \vunit_{\kappa_{i-1}})^n \ 
\REST \ \NTF \ 
(\vec{d}_{\kappa_{i-1}})^{n+1} \ 
\vunit_{\kappa_{i}}
\:.
\qedhere
\end{align*}
\end{proof}

\begin{lemma}
\label{lem:completenessF}
For every ordinal $\alpha < \omega^{d}$ and every $n \in \Nat$, $G_d$
has a computation $\tuple{n,0,\alpha} \xrightarrow{\NTF}
\tuple{F_{\alpha}(n),0,\alpha}$.
\end{lemma}
\begin{proof}
  We first observe that $G_d$ has a computation
  $\tuple{0,n,\alpha} \xrightarrow{\REST} \tuple{n,0,\alpha}$
  for every ordinal $\alpha < \omega^{d}$ and every $n \in \Nat$.
  This computation exists because $G_d$ admits derivation \eqref{derivPop},
  i.e.,
  $\REST \stepstar (\vunit_r \vec{d}_{\overline{r}})^n$.
    We now prove the lemma by induction on $\alpha$.

  \smallskip

  For the base case $\alpha = 0$, we use derivation \eqref{derivF0},
  i.e., $\NTF \stepstar \vunit_{r}$, yielding the following computation:
  \[
    \T{n}{0}{0} \xrightarrow{\NTF}
    \T{n+1}{0}{0} = \Ts{F_0(n)}{0}{0}
    \:.
  \]

  In the case of a successor ordinal $\alpha=\beta + 1$, we use derivation
  \eqref{derivF+1}, i.e.,
  $\NTF \stepstar \vec{d}_{\kappa_0} \ (\vec{d}_r \vunit_{\overline{r}})^n \ \REST \ (\NTF)^{n+1} \  \vunit_{\kappa_0}$,
  leading to the following computation
  (recall that $F_{\alpha}(n) = F_{\beta}^{n+1}(n)$ by Eq.~\eqref{eq-F1}):
  \[
    \T{n}{0}{\alpha}
    \xrightarrow{\vec{d}_{\kappa_0}}
    \T{n}{0}{\beta}
    \xrightarrow{(\vec{d}_r \vunit_{\overline{r}})^n}
    \T{0}{n}{\beta}
    \xrightarrow{\REST}
    \T{n}{0}{\beta}
    \xrightarrow[\text{ind.\ hyp.}]{\NTF^{n+1}}
    \Ts{F_{\beta}^{n+1}(n)}{0}{\beta}
    \xrightarrow{\vunit_{\kappa_0}}
    \Ts{F_{\beta}^{n+1}(n)}{0}{\beta+1}
    =\Ts{F_{\alpha}(n)}{0}{\alpha}
    \:.
  \]

  Finally, in the case of a limit ordinal $\alpha = \lambda$, say of the form
  $\lambda = \gamma + \omega^{i}$ where $0<i<d$, we use derivation
  \eqref{derivFlim}, i.e.,
  $\NTF \stepstar \vec{d}_{\kappa_i} \vunit_{\kappa_{i-1}} \ (\vec{d}_r \vunit_{\overline{r}} \vunit_{\kappa_{i-1}})^n \ \REST \ \NTF \ (\vec{d}_{\kappa_{i-1}})^{n+1} \ \vunit_{\kappa_{i}}$.
  Before inspecting the computation below, note
  that if $\lambda=\gamma+\omega^i$ is represented by the values in
  $\kappa_0, \ldots, \kappa_{d-1}$, then one obtains a representation
  for $\gamma$ by decrementing $\kappa_i$.  Then, by incrementing $(n+1)$ times
  $\kappa_{i-1}$, one obtains $\gamma+\omega^{i-1}\cdot(n+1)$ which is
  $\lambda(n)$ by \eqref{eq-L1}.
  This leads to the following computation
  (recall that $F_{\lambda}(n) = F_{\lambda(n)}(n)$ by Eq.~\eqref{eq-FL}):
  \begin{align*}
    \T{n}{0}{\lambda}
    \xrightarrow{\vec{d}_{\kappa_{i}}}
    \T{n}{0}{\gamma}
    \xrightarrow{\vunit_{\kappa_{i-1}} (\vec{d}_r \vunit_{\overline{r}} \vunit_{\kappa_{i-1}})^n}
    \T{0}{n}{\lambda(n)}
    \xrightarrow{\REST}
    \T{n}{0}{\lambda(n)}
    \xrightarrow[\text{ind.\ hyp.}]{\NTF}
    \Ts{F_{\lambda(n)}(n)}{0}{\lambda(n)}
    \:,\text{ and}
    \\
    \Ts{F_{\lambda(n)}(n)}{0}{\lambda(n)}
    =
    \Ts{F_{\lambda}(n)}{0}{\lambda(n)}
    \xrightarrow{(\vec{d}_{\kappa_{i-1}})^{n+1}}
    \Ts{F_{\lambda}(n)}{0}{\gamma}
    \xrightarrow{\vunit_{\kappa_{i}}}
    \Ts{F_{\lambda}(n)}{0}{\lambda}
    \:.
  \end{align*}

  In all three cases,
  $G_d$ has a computation $\tuple{n,0,\alpha} \xrightarrow{\NTF}
  \tuple{F_{\alpha}(n),0,\alpha}$.
\end{proof}

We showed in Lemma~\ref{lem:completenessF}
that there are computations of $G_d$ that end in
$F_{\alpha}(n)$. This corresponds to the completeness of weak
computers.  We will now show the safety part, i.e., that no successful
computation of $G_d$ may reach a value greater than $F_{\alpha}(n)$.
\begin{lemma}
\label{lem:safetyGd}
For all $n,n',m,m'\in\Nat$, $\alpha,\alpha'<\omega^d$, and $0<i<d$, the
following hold:
\begin{align}
\label{safe-F}
\tuple{n,m,\alpha}\xrightarrow{\NTF}\tuple{n',m',\alpha'}
&\implies \alpha'=\alpha \land n'+m'\leq F_{\alpha}(n+m)
\:,
\\
\label{safe-Rec}
\tuple{n,m,\alpha}\xrightarrow{\REC}\tuple{n',m',\alpha'}
&\implies \alpha'=\alpha \land n'+m'\leq F^n_{\alpha}(n+m)
\:,
\\
\label{safe-Pop}
\tuple{n,m,\alpha} \xrightarrow{\REST} \tuple{n',m',\alpha'}
&\implies \alpha' = \alpha \land n'+m' = n+m
\:,
\\
\label{safe-Lim}
\tuple{n,m,\alpha}\xrightarrow{\LIM_i}\tuple{n',m',\alpha'}
&\implies \alpha'=\alpha \land n'+m'\leq F_{\alpha\oplus(\omega^{i-1}\cdot n)}(n+m)
\:.
\end{align}
\end{lemma}
\begin{proof}
By structural induction on the flow trees witnessing the transitions.
\begin{description}[leftmargin=*]

\item[Top rule is \eqref{ruleF0} $\NTF\rightarrow \vunit_{r}$] Then
  the flow tree has the following shape (in this and following
  illustrations, \emph{we only display the top node of each immediate
  subtree} of the flow tree under consideration):
\begin{center}
  \begin{tikzpicture}[auto, node distance=2em,>=stealth]
    \node (a) {$\T{n}{m}{\alpha} \xrightarrow{\NTF} \T{n'}{m'}{\alpha'}$};
    \node (a1) [below=of a] {$\T{n}{m}{\alpha} \xrightarrow{\vunit_{r}} \T{n'}{m'}{\alpha'}$};

    \draw[-,thick] (a) to (a1);
  \end{tikzpicture}
\end{center}
Using action $\vunit_{r}$ in the subtree implies $n'=n+1$, and also
$\alpha'=\alpha$ and $m'=m$.  With \eqref{eq-F-exp}, we deduce
$n'+m'=n+m+1\leq F_\alpha(n+m)$ as required by \eqref{safe-F}.

\item[Top rule is \eqref{ruleF+1} $\NTF\rightarrow \vec{d}_{\kappa_0}
\ \REC \ \NTF \ \vunit_{\kappa_0}$] We note that the first action,
  $\vec{d}_{\kappa_0}$, can only be fired if $\alpha$ is a successor
  ordinal $\beta+1$.  Then decrementing $\kappa_0$ transforms $\alpha$
  into $\beta$, and the flow tree has the following form.
\begin{center}
  \begin{tikzpicture}[auto, node distance=2em,>=stealth]
    \node (b) {$\T{n}{m}{\alpha} \xrightarrow{\NTF} \T{n'}{m'}{\beta'+1}$};

    \node (b1) [below=of b, xshift=-11em] {$\T{n}{m}{\alpha} \xrightarrow{\vec{d}_{\kappa_0}} \T{n}{m}{\beta}$};
    \node (b21) [below=of b, xshift=-3.5em, yshift=0em] {$\T{n}{m}{\beta} \xrightarrow{\REC} \T{n_{1}}{m_{1}}{\beta_{1}}$};
    \node (b22) [below=of b, xshift=3.5em, yshift=0em] {$\T{n_{1}}{m_{1}}{\beta_{1}} \xrightarrow{\NTF} \T{n'}{m'}{\beta'}$};
    \node (b3) [below=of b, xshift=11em] {$\T{n'}{m'}{\beta'} \xrightarrow{\vunit_{\kappa_0}} \T{n'}{m'}{\beta'+1}$};

    \draw[-,thick] (b) to (b1);
    \draw[-,thick] (b) to (b3);
    \draw[-,thick] (b) to (b21);
    \draw[-,thick] (b) to (b22);
  \end{tikzpicture}
\end{center}
Invoking the induction hypothesis on the second and third subtrees yields
\begin{xalignat*}{2}
\beta_1&=\beta\:,
&
n_1+m_1&\leq F^n_{\beta}(n+m)\:,
\\
\beta'&=\beta_1\:,
&
n'+m'&\leq F_{\beta_1}(n_1+m_1)\:.
\end{xalignat*}
Combining these results, we obtain $\beta'+1=\alpha$ as needed, and
\begin{align*}
n'+m' \leq F_{\beta_1}(n_1+m_1)
&\leq F_{\beta_1}(F^n_{\beta}(n+m))
\tag*{by \eqref{eq-F-mono}}
\\
&= F^{n+1}_{\beta}(n+m)
\\
&\leq F^{n+m+1}_{\beta}(n+m)
\tag*{by \eqref{eq-F-exp}}
\\
&= F_{\beta+1}(n+m)
= F_\alpha(n+m) \:.
\tag*{by \eqref{eq-F1}}
\end{align*}
Finally, $n'+m'\leq F_\alpha(n+m)$ as required by \eqref{safe-F}.

\item[Top rule is \eqref{ruleFlim} $\NTF \rightarrow
\vec{d}_{\kappa_{i}} \vunit_{\kappa_{i-1}} \ \LIM_{i}
\ \vec{d}_{\kappa_{i-1}} \vunit_{\kappa_{i}}$] The flow tree has the
  following form.
\begin{center}
\begin{tikzpicture}[>=stealth]
	    \node (a) {$\T{n}{m}{\alpha} \xrightarrow{\NTF} \T{n'}{m'}{\alpha'}$};
	    \node (a1) [below=of a, xshift=-10em, yshift=1em] {$\T{n}{m}{\alpha} \xrightarrow{\vec{d}_{\kappa_{i}}
\vunit_{\kappa_{i-1}}} \T{n_1}{m_1}{\alpha_1}$};

	    \node (a2) [below=of a, xshift=0em, yshift=1em] {$\T{n_1}{m_1}{\alpha_1} \xrightarrow{\LIM_i} \T{n_2}{m_2}{\alpha_2}$};

	    \node (a3) [below=of a, xshift=10em, yshift=1em] {$\T{n_2}{m_2}{\alpha_2} \xrightarrow{\vec{d}_{\kappa_{i-1}}\vunit_{\kappa_{i}}} \T{n'}{m'}{\alpha'}$};

	    \draw[thick, -] (a) to (a1);
	    \draw[thick, -] (a) to (a2);
	    \draw[thick, -] (a) to (a3);
\end{tikzpicture}
\end{center}
Firing the first two actions requires decrementing $\kappa_i$, hence
$\alpha$ is some $\alpha_0\oplus\omega^i$. After these actions, one
has $n_1=n$, $m_1=m$ and $\alpha_1=\alpha_0\oplus\omega^{i-1}$.
Similarly, the last two actions require decrementing $\kappa_{i-1}$,
hence $\alpha_2$ is some $\alpha_3\oplus\omega^{i-1}$ and one has
$n'=n_2$, $m'=m_2$ and $\alpha'=\alpha_3\oplus\omega^i$.
One obtains $\alpha_1=\alpha_2$, hence $\alpha'=\alpha$, with the
induction hypothesis, as well as
\begin{align*}
n'+m'=n_2+m_2
&\leq
F_{\alpha_1\oplus(\omega^{i-1}\cdot n_1)}(n_1+m_1)
\tag*{by ind.\ hyp.}
\\
&=
F_{(\alpha_0\oplus\omega^{i-1})\oplus(\omega^{i-1}\cdot n)}(n+m)
\\
&=
F_{\alpha_0\oplus(\omega^{i}(n))}(n+m)
\\
&\leq F_{\alpha_{0}\oplus\omega^{i}}(n+m)
\tag*{by Lemma~\ref{lem-Fmono-alphaSeq}}
\\
&= F_{\alpha}(n+m)
\:,
\end{align*}
as required by \eqref{safe-F}.

\item[Top rule is \eqref{ruleRec0} $\REC\rightarrow \REST$] Then the
  flow tree has the following form.
\begin{center}
  \begin{tikzpicture}[auto, node distance=2em,>=stealth]
    \node (b) {$\T{n}{m}{\alpha} \xrightarrow{\REC} \T{n'}{m'}{\alpha'}$};
    \node (b1) [below=of b, yshift=0em] {$\T{n}{m}{\alpha} \xrightarrow{\REST} {\T{n'}{m'}{\alpha'}}$};

    \draw[-,thick] (b) to (b1);
  \end{tikzpicture}
\end{center}
The induction hypothesis gives $\alpha'=\alpha$ and $n'+m'=n+m$.  We
deduce $n'+m'\leq F_\alpha^n(n+m)$, as required by \eqref{safe-Rec},
by invoking \eqref{eq-F-exp}.

\item[Top rule is \eqref{ruleRec1} $\REC\rightarrow
\vec{d}_r\vunit_{\overline{r}}\ \REC \ \NTF$] Then $n>0$ and the
  flow tree has the following form.
\begin{center}
  \begin{tikzpicture}[auto, node distance=2em,>=stealth]
    \node (b) {$\T{n}{m}{\alpha} \xrightarrow{\REC} \T{n'}{m'}{\beta}$};

    \node (b1) [below=of b, xshift=-9em] {$\T{n}{m}{\alpha} \xrightarrow{\vec{d}_{r} \vunit_{\overline{r}}} \T{n-1}{m+1}{\alpha}$};
    \node (b2) [below=of b, yshift=-0em] {$\T{n-1}{m+1}{\alpha} \xrightarrow{\REC} \T{n_{1}}{m_{1}}{\beta_{1}}$};
    \node (b3) [below=of b, xshift=9em] {$\T{n_{1}}{m_{1}}{\beta_{1}} \xrightarrow{\NTF} \T{n'}{m'}{\beta}$};

    \draw[-,thick] (b) to (b1);
    \draw[-,thick] (b) to (b2);
    \draw[-,thick] (b) to (b3);
  \end{tikzpicture}
\end{center}
Here we can use the induction hypothesis on the second subtree, yielding
\begin{xalignat*}{2}
\beta_1&=\alpha\:,
& n_1+m_1&\leq F_\alpha^{n-1}(n-1+m+1)=F_\alpha^{n-1}(n+m)
\:,
\end{xalignat*}
and on the third subtree, yielding
\begin{xalignat*}{2}
\beta&=\beta_1\:,
& n'+m'&\leq F_{\beta_1}(n_1+m_1) = F_{\alpha}(n_1+m_1) \:.
\end{xalignat*}
Combining these and invoking \eqref{eq-F-mono}, yields $\beta=\alpha$
and $n'+m'\leq F_\alpha(n_1+m_1)\leq F_\alpha(F_\alpha^{n-1}(n+m))=
F_\alpha^n(n+m)$ as required by \eqref{safe-Rec}.

\item[Top rule is \eqref{ruleRest0}
$\REST\rightarrow\varepsilon$] Then the flow tree 
$(n,m,\alpha)\xrightarrow{\REST}(n',m',\alpha')$
is a leaf, entailing  $\alpha'=\alpha$ and $n'+m'=n+m$ as required by
\eqref{safe-Pop}.

\item[Top rule is \eqref{ruleRest1} $\REST \rightarrow
\vunit_r \vec{d}_{\overline{r}} \ \REST$] Then the flow tree has the form.
\begin{center}
  \begin{tikzpicture}[auto, node distance=2em,>=stealth]
    \node (b)  {$\T{n}{m}{\alpha} \xrightarrow{\REST} \T{n'}{m'}{\alpha'}$};
    \node (b1) [below=of b, xshift=-5em] {$\T{n}{m}{\alpha} \xrightarrow{\vunit_{r} \vec{d}_{\overline{r}} }\T{n+1}{m-1}{\alpha}$};
    \node (b2) [below=of b, xshift=5em] {$\T{n+1}{m-1}{\alpha} \xrightarrow{\REST} \T{n'}{m'}{\alpha'}$};

    \draw[-,thick] (b) to (b1);
    \draw[-,thick] (b) to (b2);
  \end{tikzpicture}
\end{center}
The induction hypothesis on the second subtree
gives $\alpha'=\alpha$ and $n'+m'=n+1+m-1=n+m$ as required by
\eqref{safe-Pop}.

\item[Top rule is \eqref{ruleLim0} $\LIM_{i}\rightarrow \REST \ \NTF$] The flow tree has the following form.
\begin{center}
    \begin{tikzpicture}[>=stealth]
	\node (a) {$\T{n}{m}{\alpha} \xrightarrow{\LIM_{i}} \T{n'}{m'}{\alpha'}$};

	\node (a1) [below=of a, xshift=-7em, yshift=1em] {$\T{n}{m}{\alpha} \xrightarrow{\REST} \T{n_{1}}{m_{1}}{\alpha_{1}}$};
	\node (a2) [below=of a, xshift=7em, yshift=1em] {$\T{n_{1}}{m_{1}}{\alpha_{1}} \xrightarrow{\NTF} \T{n'}{m'}{\alpha'}$};

	\draw[thick, -] (a) to (a1);
	\draw[thick, -] (a) to (a2);
    \end{tikzpicture}
\end{center}
On these subtrees, the induction hypothesis yields $\alpha'=\alpha_{1}
= \alpha$ and $n_{1} + m_{1} = n + m$. Furthermore we have
\begin{align*}
n' + m' &\leq F_{\alpha_1}(n_{1} + m_{1})
\tag*{by ind.\ hyp.}
\\
& \leq F_{\alpha\oplus(\omega^{i-1}\cdot n)}(n+m)
\:,
\tag*{by Lemma~\ref{lem-Fmono-alpha}}
\end{align*}
as required by
\eqref{safe-Lim}.

\item[Top rule is \eqref{ruleLim1} $\LIM_{i} \rightarrow \vec{d}_{r}
\vunit_{\overline{r}} \vunit_{\kappa_{i-1}} \ \LIM_{i} \ \vec{d}_{\kappa_{i-1}}$] The flow tree has the following form.
\begin{center}
	\begin{tikzpicture}[>=stealth]
	    \node (a) {$\T{n}{m}{\alpha} \xrightarrow{\LIM_{i}} \T{n'}{m'}{\alpha'}$};
	    \node (a1) [below=of a, xshift=-10em, yshift=1em] {$\T{n}{m}{\alpha} \xrightarrow{\vec{d}_{r} \vunit_{\overline{r}} \vunit_{\kappa_{i-1}}} \T{n_1}{m_1}{\alpha_1}$};
	    \node (a2) [below=of a, xshift=0em, yshift=1em] {$\T{n_1}{m_1}{\alpha_1} \xrightarrow{\LIM_{i}} \T{n_2}{m_2}{\alpha_2}$};
	    \node (a3) [below=of a, xshift=10em, yshift=1em] {$\T{n_2}{m_2}{\alpha_2} \xrightarrow{\vec{d}_{\kappa_{i-1}}} \T{n'}{m'}{\alpha'}$};

	    \draw[thick, -] (a) to (a1);
	    \draw[thick, -] (a) to (a2);
	    \draw[thick, -] (a) to (a3);
	\end{tikzpicture}
\end{center}
With its three actions, the first subtree implies
\begin{xalignat*}{3}
  n_1&=n-1 \:,
& m_1&=m+1 \:,
&\alpha_1&=\alpha\oplus\omega^{i-1} \:.
\end{xalignat*}
Similarly, the last subtree yields
\begin{xalignat*}{3}
  n'&=n_2 \:,
& m'&=m_2 \:,
&\alpha_2&=\alpha'\oplus\omega^{i-1}\:.
\end{xalignat*}
With the second subtree, the induction hypothesis yields
$\alpha_2=\alpha_1$ and $n_2+m_2\leq
F_{\alpha_1\oplus(\omega^{i-1}\cdot n_1)}(n_1+m_1)$.  Combining these
gives $\alpha'=\alpha$ and $n'+m'=n_2+m_2\leq
F_{\alpha_1\oplus(\omega^{i-1}\cdot n_1)}(n_1+m_1) =
F_{(\alpha\oplus\omega^{i-1})\oplus(\omega^{i-1}\cdot (n-1))}(n+m) =
F_{\alpha\oplus(\omega^{i-1}\cdot n)}(n+m)$ as required by
\eqref{safe-Lim}.
\qedhere
\end{description}
\end{proof}

Now, for any ordinal $\alpha<\omega^d$, we may extend $G_d$ and obtain
a GVAS $G_{F_\alpha}$ that weakly computes $F_\alpha$.  This new GVAS
inherits the counters, actions, non-terminals and rules of $G_d$. It
furthermore includes two additional non-terminals, $\NTS$ and
$\REST'$, and associated rules. The start symbol will be $\NTS$ and,
if $\alpha$'s CNF is $\sum_{i=d-1}^{0}\omega^i\cdot c_i$, the extra
rules are:
\begin{xalignat}{1}
\NTS   & \rightarrow
                \vunit_{\kappa_0}^{c_0}
                \vunit_{\kappa_1}^{c_1}
                \cdots
                \vunit_{\kappa_{d-1}}^{c_{d-1}}
                \ \NTF
                \ \REST'        \:,
\label{ruleS}\tag{R0}
\\
\REST' & \rightarrow \varepsilon \:,
\label{ruleTrans0}\tag{R10}
\\
\REST' & \rightarrow \vec{d}_{r}\vunit_{\overline{r}} \ \REST' \:.
\label{ruleTrans1}\tag{R11}
\end{xalignat}
It is clear that, since there are no new rules for the non-terminals
inherited from $G_d$, the properties stated in
Lemmas~\ref{lem:completeness-derivations} to~\ref{lem:safetyGd} hold for
$G_{F_\alpha}$ as they hold for $G_d$. Note also that $\REST'$ behaves
as $\REST$ but exchanging the roles of $r$ and $\overline{r}$.
\begin{lemma}
\label{lem-GFa-WC}
$G_{F_\alpha}$ weakly computes $F_{\alpha}$.
\end{lemma}
\begin{proof}
We start with the completeness part of Definition~\ref{def-wpn}.  For
this it is needed to show that, for any $n\in\Nat$, $G_{F_\alpha}$ has
a computation of the form $(n,0,\vzero_d)\xrightarrow{\NTS}
(n',F_{\alpha}(n),\vec{e})$.  For this we use the rule \eqref{ruleS}
and exhibit the following computation:
\[
(n,0,\vzero_d)
\xrightarrow{\vunit_{\kappa_0}^{c_0}
                \vunit_{\kappa_1}^{c_1}
                \cdots
                \vunit_{\kappa_{d-1}}^{c_{d-1}}}
(n,0,\alpha)
\xrightarrow{\NTF}
(F_{\alpha}(n),0,\alpha)
\xrightarrow{\REST'}
(0,F_{\alpha}(n),\alpha)
\:.
\]
The first part of that computation just relies on our convention for
reading $\kappa_{d-1},\ldots,\kappa_0$ as the encoding of an ordinal,
the second (crucial) part is given by Lemma~\ref{lem:completenessF},
and the last part is by an analog of
derivation \eqref{derivPop} for $\REST'$.

For the safety part, we consider an arbitrary computation of the form
$(n,0,\vzero_d)\xrightarrow{\NTS}(n',r,\vec{e})$. The only rule for $\NTS$ is
\eqref{ruleS}, so there must exist some steps of the form
\[
(n,0,\vzero_d)
\xrightarrow{\vunit_{\kappa_0}^{c_0}
                \vunit_{\kappa_1}^{c_1}
                \cdots
                \vunit_{\kappa_{d-1}}^{c_{d-1}}}
(n,0,\alpha)
\xrightarrow{\NTF}
(n',m',\alpha')
\xrightarrow{\REST'}
(n'',r,\vec{e})
\:.
\]
Necessarily, this satisfies $n'+m'\leq F_{\alpha}(n)$ by
\eqref{safe-F}. And since $\REST'$ behaves like $\REST$, we have
$n''+r=n'+m'$, as in \eqref{safe-Pop}. All this entails $r\leq
F_{\alpha}(n)$ as required by \eqref{eq-sa}.
\end{proof}

 \section{Concluding Remarks}
\label{sec-concl}

We proved that Grammar-controlled VASes or Pushdown VASes cannot
weakly compute numerical functions that are sublinear. This was
recently shown for plain VASes~\cite{LerSch-rp2014}. We also proved
that GVASes can weakly compute the fast-growing functions $F_\alpha$ for
all $\alpha<\omega^\omega$ while VASes can only weakly compute
$F_\alpha$ for $\alpha<\omega$.

This research is motivated by verification questions for
well-structured systems, in particular VASes and their extensions. In this
area, weakly computable functions have traditionally been used to
prove hardness results. Recent hardness proofs for well-structured
systems crucially rely on the ability to weakly compute both
fast-growing and slow-growing functions.

This work raises some new questions that are left for future work,
including whether GVASes can weakly compute $F_{\omega^\omega}$ and whether
slow-growing functions can be weakly computed in other VAS extensions like the VASes with nested zero-tests of~\cite{Reinhardt08}.

Another open question is the decidability of the boundedness problem
for GVASes. Boundedness is decidable for PVASes~\cite{leroux2014} but
the two problems do not coincide: on the one hand, in GVASes we only
consider sequences of actions that are the yields of complete
derivation trees of a grammar, corresponding to configurations in PVAS
that have empty stack content; on the other hand unboundedness in
PVASes can come from unbounded counters or unbounded stack, while in
GVASes only counters are measured.  Indeed, the counter-boundedness
problem for PVASes reduces to the boundedness problem for GVASes and
is still open, while the stack-boundedness problem was shown decidable
in~\cite{leroux2015c}.

The reachability problem for GVASes is also a source of open
problems. Recently, the complexity of the reachability problem for plain VASes was proved
to be between $\mathbf{F}_3$ and
$\mathbf{F}_\omega$~\cite{leroux-stoc19,leroux-lics19} in the complexity
hierarchy set up by Schmitz~\cite{schmitz-toct2016}. Improving the
$\mathbf{F}_3$ lower bound in the case of GVASes is an open question,
as is the decidability status of the
reachability problem.

\bibliographystyle{plain}
\bibliography{wcca}

\appendix

\section{Well-quasi-ordering flow trees}
\label{app-wqo}

We now prove Lemma~\ref{lem-FG-wqo}, stating that $\leq_G$ is a
well-quasi-ordering of $F(G)$, this for any GVASS $G$.  A simple way
to prove this is to reformulate $\leq_G$ as an homeomorphic embedding
on flow trees labeled with enriched information equipped with further
labels.

With a flow tree
$t=\sigma[t_1,\ldots,t_k]$ we associate a tree $\child(t)$ made of a
root labeled with
$\langle\root(t),\root(t_1),\ldots,\root(t_k)\rangle$ and having
$\child(t_1),\ldots,\child(t_k)$ as immediate subtrees.  
The nodes of $\child(t)$ are labeled by tuples of the form
$\langle(\vc_0,X,\vc_k),(\vc_0,X_1,\vc_1),\ldots,
(\vc_{k-1},X_k,\vc_k)\rangle$ where $\vc_0,\ldots,\vc_k\in \Nat^d$
are configurations and where $X\step X_1\cdots X_k$ is a rule in $R$, or
$X\in \vec{A}$ is a terminal action and $k=0$.  Such a tuple is called an
\emph{instance} of the rule $X\step X_1\cdots X_k$ (or of the action
$X\in \vec{A}$).  Given two instances
$\lambda=\langle(\vc_0,X,\vc_k),(\vc_0,X_1,\vc_1),\ldots,
(\vc_{k-1},X_k,\vc_k)\rangle$ and
$\lambda'=\langle(\vd_0,Y,\vd_\ell),(\vd_0,Y_1,\vd_1),\ldots,
(\vd_{\ell-1},Y_\ell,\vd_\ell)\rangle$, we write $\lambda\leq\lambda'$
when $(X\step X_1\cdots X_k)$ and $(Y\step Y_1\cdots Y_\ell)$ are the same rule or
action ---entailing $k=\ell$--- and when $\vc_j\leq \vd_j$ for all
$1\leq j\leq k$. Suppose $s$ is a flow tree with immediate subtrees $s_1, \ldots,
s_k$. We write $\child(s) \sqsubseteq \child(t)$ if there is a
subtree $t'$ of $t$ with immediate subtrees $t_1, \ldots, t_l$ such that
$\root(\child(s)) \le \root(\child(t'))$ (entailing $k = l$) and inductively
$\child(s_j) \sqsubseteq \child(t_j)$ for all $1 \le j \le k$. It is easy to see
that on derived trees of flow trees, $\sqsubseteq$ coincides with the standard homeomorphic
embedding of labeled trees.
\begin{lemma}
\label{lem-FG-kruskal}
$s\leq_G t$ if, and only if, $\root(s)\leq\root(t) \land
\child(s)\sqsubseteq\child(t)$.
\end{lemma}
\begin{proof} We assume
$s=\sigma[s_1,\ldots,s_k]$, $t=\theta[t_1,\ldots,t_\ell]$ and prove
the claim by structural induction.\\
\noindent
${\implies}$: Assume $s\leq_G t$. Thus $\root(s)\leq\root(t)$ and, by
definition of $\leq_G$, $t$ contains a subtree
$t'=\theta'[t'_1,\ldots,t'_k]$ with
\[
\sigma\leq\theta' \land s_1\leq_G t'_1 \land \cdots \land s_k\leq_G
t'_k \:.
\]
Now this entails $\root(s)\leq\root(t')$ and
$\root(s_i)\leq\root(t'_i)$ for all $i$, as well as (by ind.\ hyp.)
$\child(s_i)\sqsubseteq\child(t'_i)$ for all $i$.  Hence
$\child(s)\sqsubseteq\child(t')$, entailing $\child(s)\sqsubseteq\child(t)$.

\noindent
${\impliedby}$: We assume $\root(s)\leq\root(t)$ and
$\child(s)\sqsubseteq\child(t)$. Hence there is a subtree $t'$ of $t$ with
immediate subtrees $t'_1, \ldots, t'_k$ such that
$\root(\child(s))\leq\root(\child(t'))$ and
$\child(s_i)\sqsubseteq \child(t'_i)$ for all $i=1,\ldots,k$. From
$\root(\child(s))\leq\root(\child(t'))$, we infer that $\root(s_i) \leq
\root(t'_i)$ for all $i=1, \ldots,k$.
Now one witnesses $s\leq_G t$ by observing that
$s_i\leq_G t_i$ by ind.\ hyp.
\end{proof}
Since the instances of rules are well-quasi-ordered by $\leq$ (there
are only finitely many rules), $\sqsubseteq$ is a well-quasi-ordering
by Kruskal's Tree Theorem~\cite{kruskal60}. 
With Lemma~\ref{lem-FG-kruskal} we immediately
infer that $\leq_G$ is a well-quasi-ordering.

\section{Monotonicity for Fast-Growing functions}
\label{app-mon-fg}

We give detailed proofs for the two monotonicity lemmas stated after the
definitions of the $F_\alpha$ functions in section~\ref{sec-hypack}.

\gdef\thesection{\ref*{sec-hypack}}
\setcounter{thm}{1}

\begin{lemma}
For any ordinals $\alpha,\alpha'<\epsilon_0$ and any $n\in\Nat$,
$F_{\alpha}(n) \leq F_{\alpha\oplus\alpha'}(n)$.
\end{lemma}
\begin{proof}
By induction on $\alpha'$, then on $\alpha$.  We first observe that, if
the claim holds for some given $\alpha$ and $\alpha'$, then it entails
$F^m_{\alpha}(n) \leq F^m_{\alpha\oplus\alpha'}(n)$ for any $m>0$
as a consequence of monotonicity, i.e., \eqref{eq-F-mono}.

We now consider several cases for $\alpha$ and $\alpha'$:
\begin{itemize}
\item
If $\alpha'=0$, then $\alpha\oplus\alpha'=\alpha$ and the claim holds
trivially.
\item
If $\alpha=0$ then the claim becomes $F_0(n)\leq F_{\alpha'}(n)$, which
holds since $F_0(n)=n+1$ by \eqref{eq-F0} and $n+1\leq F_{\alpha'}(n)$ by
\eqref{eq-F-exp}.
\item
If $\alpha'=\beta'+1$ is a successor then $\alpha\oplus\alpha'$ is
$(\alpha\oplus\beta')+1$ and we have $F_{\alpha}(n)\leq
F_{\alpha\oplus\beta'}(n)$ by ind.\ hyp., $\leq
F_{\alpha\oplus\beta'}^{n+1}(n)$ by \eqref{eq-F-exp}, $=
F_{\alpha\oplus\beta'+1}(n)=F_{\alpha\oplus\alpha'}(n)$, and we are
done.
\item
If $\alpha=\beta+1$ is a successor then $\alpha\oplus\alpha'$ is
$(\beta\oplus\alpha')+1$ and we have $F_\alpha(n)
=F_{\beta}^{n+1}(n) \leq F_{\beta\oplus\alpha'}^{n+1}(n)$ by
ind.\ hyp., $=F_{\alpha\oplus\alpha'}(n)$.
\item
The only remaining possibility is that both $\alpha$ and $\alpha'$ are
limit ordinals. Then $(\alpha\oplus\alpha')(n)$ is
$\alpha\oplus\alpha'(n)$ or $\alpha(n)\oplus\alpha'$, depending on which
limit has the CNF with smallest last summand.  In the first case we
have $F_\alpha(n)\leq F_{\alpha\oplus\alpha'(n)}(n)$ by
ind.\ hyp.\ since $\alpha'(n)<\alpha'$,
$=F_{(\alpha\oplus\alpha')(n)}(n)=F_{\alpha\oplus\alpha'}(n)$.  In the
second case we have $F_\alpha(n)=F_{\alpha(n)}(n)\leq
F_{\alpha(n)\oplus\alpha'}(n)$ by ind.\ hyp.\ since $\alpha(n)<\alpha$,
$=F_{(\alpha\oplus\alpha')(n)}(n)=F_{\alpha\oplus\alpha'}(n)$.
\qedhere
\end{itemize}
\end{proof}

\begin{lemma}
For any ordinal $\alpha<\epsilon_0$ and limit ordinal $\lambda<\omega^\omega$,
for any $m, n\in\Nat$, if $m \leq n$ then
$F_{\alpha\oplus\lambda(m)}(n)
\leq
F_{\alpha\oplus\lambda}(n)$.
\end{lemma}
\begin{proof}
Let us decompose $\lambda$ under the form
$\lambda=\delta+\omega^{k+1}$ so that $\lambda(m)=\delta
+\omega^{k}\cdot (m+1)$.  We prove the lemma by induction on $\alpha$.
We first observe that,
if the claim holds for some given $\alpha$, then
$F^p_{\alpha\oplus\lambda(m)}(n)
\leq
F^p_{\alpha\oplus\lambda}(n)$
for every $p > 0$ and every $m, n\in\Nat$ such that $m \leq n$.
This observation, which is easily proved by induction on $p$,
is a consequence of strict expansivity and monotonicity, i.e.,
\eqref{eq-F-exp} and \eqref{eq-F-mono}, respectively.
\begin{itemize}
\item
If $\alpha=0$ we have
$F_{\alpha\oplus\lambda(m)}(n)
= F_{\lambda(m)}(n)
\leq F_{\lambda(n)}(n)$ by
Lemma~\ref{lem-Fmono-alpha} since $\lambda(n)=\lambda(m)\oplus
\omega^k\cdot (n-m)$,
$= F_{\lambda}(n)
= F_{\alpha\oplus\lambda}(n)$ by \eqref{eq-FL}.

\item
If $\alpha=\beta+1$ is a successor, we have
$F_{\alpha\oplus\lambda(m)}(n)
= F_{(\beta\oplus\lambda(m))+1}(n)
= F^{n+1}_{\beta\oplus\lambda(m)}(n)$ by \eqref{eq-F1},
$\leq F^{n+1}_{\beta\oplus\lambda}(n)$ by ind.\ hyp.,
$= F_{(\beta\oplus\lambda)+1}(n) = F_{\alpha\oplus\lambda}(n)$ again by
\eqref{eq-F1}.

\item
If $\alpha=\gamma+\omega^\beta$ is a limit, then $\alpha\oplus\lambda$ is
a limit too and we can compare $\omega^\beta$ and $\omega^{k+1}$, the last summands of $\alpha$ and
$\lambda$.
There are two subcases:
\begin{itemize}
\item
If $0 < \beta \leq k$ then $(\alpha\oplus \lambda)(n) = \alpha(n)\oplus \lambda$.
Moreover, $\alpha\oplus\lambda(m)$ is a limit since $0 < k$,
and $(\alpha\oplus\lambda(m))(n)=\alpha(n)\oplus\lambda(m)$.
We deduce
$F_{\alpha\oplus\lambda(m)}(n) = F_{\alpha(n)\oplus\lambda(m)}(n)$ by \eqref{eq-FL},
$\leq F_{\alpha(n)\oplus\lambda}(n)$ by ind.\ hyp.,
$= F_{(\alpha\oplus\lambda)(n)}(n) = F_{\alpha\oplus\lambda}(n)$ again by \eqref{eq-FL}.
\item
If $k+1 \leq \beta$ then $(\alpha\oplus \lambda)(n) = \alpha\oplus
\lambda(n)$.
We deduce
$F_{\alpha\oplus\lambda(m)}(n)
\leq F_{\alpha\oplus\lambda(n)}(n)$ by Lemma~\ref{lem-Fmono-alpha}
since $m\leq n$,
$= F_{(\alpha\oplus\lambda)(n)}(n)
= F_{\alpha\oplus\lambda}(n)$ by \eqref{eq-FL} and we are done.
\qedhere
\end{itemize}
\end{itemize}
\end{proof}

\end{document}